\documentclass[10pt]{article}
\usepackage{graphicx} 
\usepackage{amsmath}
\usepackage{amssymb}
\usepackage{amsfonts}
\usepackage{amsthm}
\usepackage{enumitem}
\usepackage{array}
\usepackage{tikz-cd}
\usepackage{biblatex}
\usepackage{arydshln}
\newcommand{\Cls}{\mathbf{C}}
\newcommand{\Clsp}{\mathbf{C'}}
\newcommand{\Two}{\mathbf{2}}
\newcommand{\One}{\mathbf{1}}
\newcommand{\Total}{\mathbf{Total}}

\newcommand{\All}{\mathbf{All}}

\newcommand{\CKER}{\mathsf{CKE}}

\newcommand{\Cons}{\mathsf{Cons}}

\newcommand{\CKS}{\mathsf{CK}}

\newcommand{\mdom}[2]{#1^{#2}}
\newcommand{\mdomA}[1]{\mdom{A}{#1}}

\newcommand{\NP}{\mathbf{NP}}
\newcommand{\cNP}{\mathbf{coNP}}
\newcommand{\ER}{\exists \mathbb{R}}
\newcommand{\EC}{\exists \mathbb{C}}
\newcommand{\ECconj}{\exists \overline{\mathbb{C}}}
\newcommand{\PSPACE}{\mathbf{PSPACE}}

\newcommand{\DEXPTIME}{\mathbf{2EXPTIME}}

\newcommand{\weakSAT}[1]{\mathbf{SAT}^{\not=0}(#1)}
\newcommand{\strongSAT}[1]{\mathbf{SAT}^{=1}(#1)}
\newcommand{\weakCNF}[1]{\mathbf{3CNFSAT}^{\not=0}(#1)}
\newcommand{\strongCNF}[1]{\mathbf{3CNFSAT}^{=1}(#1)}
\newcommand{\bothCNF}[1]{\mathbf{3CNFSAT}(#1)}

\newcommand{\strongXOR}[1]{\mathbf{XORSAT}^{=1}(#1)}
\newcommand{\bothXOR}[1]{\mathbf{XORSAT}(#1)}
\newcommand{\SAT}[1]{\mathbf{SAT}(#1)}
\newcommand{\XSAT}{\mathbf{XSAT}}
\newcommand{\QHOM}{\mathbf{QHOM}}

\newcommand{\VARSAT}{\mathbf{VARSAT}}
\newcommand{\ar}{\mathsf{ar}}
\newcommand{\Lrings}{L_{\mathsf{rings}}}

\newcommand{\setn}{[n]}

\newcommand{\setd}{[d]}
\newcommand{\setr}{[r]}

\newcommand{\impl}{\rightarrow}
\newcommand{\cimpl}{\leftarrow}
\newcommand{\bimpl}{\leftrightarrow}
\newcommand{\xor}{\oplus}

\newcommand{\sg}{\sigma}
\newcommand{\Ms}{\mathcal{M}}
\newcommand{\Ns}{\mathcal{N}}
\newcommand{\Ts}{\mathcal{T}}
\newcommand{\Vs}{\mathcal{V}}

\newcommand{\RPTwo}{\mathbb{RP}^2}
\newcommand{\RAThree}{\mathbb{R}^3}
\newcommand{\FK}{\mathbb{K}}
\newcommand{\FZ}{\mathbb{Z}}
\newcommand{\FC}{\mathbb{C}}
\newcommand{\FR}{\mathbb{R}}
\newcommand{\iu}{i}

\newcommand{\cross}{\rightthreetimes}

\newcommand{\ocomp}[1]{\overline{#1}}

\newcommand{\rank}{\mathbf{rank}}
\newcommand{\im}{\mathbf{Im}}
\newcommand{\basis}{\mathsf{basis}}

\newcommand{\res}[2]{#1|_{#2}}

\newcommand{\oext}{\circledcirc}
\newcommand{\qs}{\backslash}

\newcommand{\Var}{\mathsf{Var}}
\newcommand{\Sub}{\mathsf{Sub}}
\newcommand{\Eq}{\mathsf{Eq}}
\newcommand{\CNF}{\mathsf{CNF}}

\newcommand{\Bnd}{\mathbf{B}}
\newcommand{\Pj}{\mathbf{P}}
\newcommand{\Oj}{\mathbf{O}}
\newcommand{\Inv}{\mathbf{I}}
\newcommand{\Hilb}{\mathcal{H}}

\newcommand{\BA}{\mathbf{BA}}
\newcommand{\pBA}{\mathbf{pBA}}

\newcommand{\Qd}{\mathbf{Q}_d}

\makeatletter
\def\namedlabel#1#2{\begingroup
   \def\@currentlabel{#2}%
   \label{#1}\endgroup
}
\makeatother

\theoremstyle{plain}
\newtheorem{theorem}{Theorem}[section]
\newtheorem{proposition}[theorem]{Proposition}
\newtheorem{corollary}[theorem]{Corollary}
\newtheorem{lemma}[theorem]{Lemma}

\theoremstyle{definition}
\newtheorem{definition}[theorem]{Definition}

\theoremstyle{remark}
\newtheorem{remark}[theorem]{Remark}

\title{Complexity of Satisfiability in Kochen-Specker Partial Boolean Algebras}
\author{Anuj Dawar and Nihil Shah}
\date{\today}

\begin{document}
\maketitle
\begin{abstract}
The Kochen-Specker no-go theorem established that hidden-variable theories in quantum mechanics necessarily admit contextuality.
This theorem is formally stated in terms of the partial Boolean algebra structure of projectors on a Hilbert space.
Each partial Boolean algebra provides a semantics for interpreting propositional logic.
In this paper, we examine the complexity of propositional satisfiablity for various classes of partial Boolean algebras.
We first show that the satisfiability problem for the class of non-trivial partial Boolean algebras is $\NP$-complete.
Next, we consider the satisfiability problem for the class of partial Boolean algebras arising from projectors on finite dimensional Hilbert spaces.
For real Hilbert spaces of dimension greater 2 and any complex Hilbert spaces of dimension greater than 3, we demonstrate that the satisfiablity problem is complete for the existential theory of the reals.
Interestingly, the proofs of these results make use of Kochen-Specker sets as gadgets.
As a corollary, we conclude that deciding quantum homomorphism in these fixed dimensions are also complete for the existential theory of the reals.
Finally, we show that the satisfiability problems for the class of all Hilbert spaces and all finite-dimensional Hilbert spaces is undecidable.
\end{abstract}

\section{Introduction}
The topic of this paper lies at the intersection of one notable aspect in each of twentieth century physics, mathematics, and computer science.
In physics, this aspect is the inherent contextuality in dealing with quantum physical systems as exemplified by the Kochen-Specker no-go theorem.
In mathematics, this aspect is the algebraic tradition of logic, where we study the partiality relaxation, arising from Kochen and Specker's work, of the standard Boolean algebra semantics for propositional logic.
In computer science, this aspect is understanding how the statement of Cook-Levin's theorem, demonstrating that propositional satisfiability is $\NP$-complete, changes under this partiality relaxation.

The Kochen-Specker theorem, appearing in~\cite{hiddenVariableQM} proved the impossibility of embedding quantum theory into any classical physical theory possibly involving hidden variables.
This resolved an open question that had plagued quantum theory since its inception.
A key step in formulating a mathematical resolution to the question was to establish a necessary condition on observables that would hold if quantum theory was embeddable into a classical theory.
Observables, in standard quantum theory, are represented as operators on the Hilbert space $\Hilb$ of a quantum system's states.
The necessary condition was that the partial algebra of operators $\Bnd(\Hilb)$, where algebraic operations are only defined on pairs of operators which commute, could be embedded into a single commutative algebra $\mathcal{A}$.
This formalisation addresses the question since only observables which are represented by commuting operators are physically commeasurable, i.e.\ can be performed together.
An immediate observation made by Kochen and Specker was that any embedding from $\Bnd(\Hilb)$ to a commutative algebra would induce a homomorphism of the partial Boolean algebra of idempotent elements $\Pj(\Hilb)$ into the two-element Boolean algebra.
Thus, the Kochen-Specker theorem was proved by showing that for any Hilbert space $\Hilb$ of dimension $\geq 3$, there is no homomorphism from $\Pj(\Hilb)$ into the two-element Boolean algebra.

The idempotent elements of $\Bnd(\Hilb)$ forming the partial Boolean algebra $\Pj(\Hilb)$ are the projectors of the Hilbert space $\Hilb$.
Projectors, or equivalently their associated subspace, can be identified with yes/no questions, i.e.\ propositions, about the quantum system corresponding to $\Hilb$.
The ``and'' and ``or'' operations $\Pj(\Hilb)$ are only defined on projectors which commute.
This respects the postulate that only commuting observables can be jointly measured, i.e.\ are commeasurable.

More generally, a partial Boolean algebra is equipped with a binary, reflexive, symmetric commeasurability relation which determines the domain of definition for the operations.
Since these operations are partial, the usual notion of substitution of variables in a term is adapted to \emph{meaningful substitution}.
Thus, each partial Boolean algebra provides an alternative semantics of propositional logic via this notion of meaningful substitution.

In a paper dedicated to Specker, Kochen~\cite{quantumReconstruction} argues that since these semantics restrict to physically meaningful substitutions, the partial Boolean algebra approach can be used to reconstruct significant features of reasoning in quantum mechanics.
Moreover, recent papers~\cite{noncommutativityAsColimit,logicContextuality} have been reviving research in this partial Boolean algebra semantics for propositional logic in order to provide a fresh perspective on `quantum logic'.
This would contrast with the historical Birkoff-von Neumann approach of propositional quantum semantics in terms of orthomodular posets.
These papers highlight that a key consequence of the Kochen-Specker theorem is that there exist propositional tautologies which are unsatisfiable under this alternative semantics.
A question the authors put forth in~\cite{logicContextuality} is to obtain characterisation of all such Kochen-Specker paradoxes.

This characterisation question is the primary motivation for this paper.
The traditional mathematical logic approach to this question had been partially investigated in~\cite{logicalQuantum}.
There, Kochen and Specker exhibited a sound and complete proof system for determining which propositional formulas are valid in all non-trivial partial Boolean algebras.
By contrast, we take a computer science approach to the question by exhibiting algorithms for recognising Kochen-Specker paradoxes in various interesting classes of partial Boolean algebras.
More precisely, we work with the dual question by exhibiting algorithms for deciding if an input propositional formula is satisfiable in a given class of partial Boolean algebras.
Proving complexity bounds for such an algorithm would in fact be a natural adaptation of a staple result of theoretical computer science: the Cook-Levin theorem.
The Cook-Levin theorem established that the computational complexity of deciding whether a propositional formula is satisfiable in total Boolean algebras is complete for non-deterministic polynomial time.
Thus, the topic of this paper is to establish the complexity bounds for different classes of partial Boolean algebras.

\textbf{Outline} In Section~\ref{sec:prelim}, we review notational preliminaries and background material on partial Boolean algebras which we use throughout the rest of the paper.
Section~\ref{sec:observations} formulates the weak and strong satisfiability decision problems, and their dual validity decision problems, for arbitrary classes $\Cls$ of partial Boolean algebras.
The section also makes some general observations about these problems and concludes with a sufficient condition on $\Cls$ to guarantee $\Cls$-weak and $\Cls$-strong satisfiability problems are at least as hard as the classical satisfiability problem.
We also demonstrate the satisfiability problem in which a pBA is given as additional input is $\NP$-complete.
Section~\ref{sec:all-sat} prove that satisfiability problem for the class of all non-trivial partial Boolean algebras is $\NP$-complete.
Section~\ref{sec:kochen-specker} reviews Kochen-Specker proofs how they can be encoded as graphs and propositional formulas.
We leverage these Kochen-Specker proofs in the next section by using them as gadgets in a hardness reduction.
In Section~\ref{sec:fixed-quantum-sat}, we show that for any fixed dimension $d \geq 3$, the satisfiability problem for $\Pj(\FR^d)$ is complete for the existential theory of the reals.
Similarly, we show that for any fixed dimension $d \geq 4$, the satisfiability problem for $\Pj(\FC^d)$ is complete for the existential theory of the reals.
This section also establishes the for dimension $d = 1 \text{ or } 2$, the satisfiability problem is $\NP$-complete.
Section~\ref{sec:quantum-sat} show that the satisfiablity problems for pBAs arising from the class of finite dimensional and the class of all Hilbert spaces are undecidable.

\textbf{Related past work} Variations of our results have been studied in the previous literature.
In particular, the projectors over real and complex Hilbert space $\Hilb$ can also be equipped with a orthomodular lattice structure $\Oj(\Hilb)$ rather than the partial Boolean algebra structure $\Pj(\Hilb)$.
Orthomodular lattices satisfy a weakening of the distributivity axiom of Boolean algebras.
The orthomodular lattice $\Oj(\Hilb)$ interprets $\vee$ and $\wedge$ as total operations which return the projector onto the union and intersection, respectively, of the input projectors' image subspaces.
In contrast to the partial operations of $\Pj(\Hilb)$, the total operations of $\Oj(\Hilb)$ do not have a simple algebraic definition in terms of operator addition and composition, see e.g.\ \cite{orthoLatticeOperations} for details.
Moreover, these operations arguably lack a physical interpretation~\cite{quantumReconstruction}.
Orthomodular lattices are the algebraic semantics used in the traditional Birkhoff-von Neumann approach to quantum logic.
It is known that  satisfiability of propositional formulas in $\Oj(\Hilb)$ is $\ER$-complete for any $d$-dimensional real or complex Hilbert space where $d \geq 3$~\cite{quantumSatOrtholatticeFixed}.
Theorem~\ref{thm:fixed-d-real-sat} and Theorem~\ref{thm:fixed-d-complex-sat} establish the same $\ER$-completeness result for $\Pj(\Hilb)$.
It has been shown that satisfiability in $\Oj(\Hilb)$ for any finite dimensional real or complex Hilbert space is undecidable~\cite{quantumSatOrtholatticeArbitrary}.
Theorem~\ref{thm:full-dim-sat} establishes the same undecidability result for $\Pj(\Hilb)$.

That there is no algorithm for deciding if a constraint satisfaction problem over a Boolean domain has a satisfying assignment of involutive operators in $\Bnd(\Hilb)$ for some finite-dimensional Hilbert space $\Hilb$ was shown in~\cite{quantumSchaefer}.
Constraint satisfaction problems over a Boolean domain can be identified with propositional formulas in conjunctive normal form and involutive operators in $\Bnd(\Hilb)$ are in bijection with projectors $\Pj(\Hilb)$.
Thus, Theorem~\ref{thm:full-dim-sat} is an algebraic reformulation of one of the undecidability results in~\cite{quantumSchaefer}.
We discuss the precise connection with Theorem~\ref{thm:full-dim-sat} in Remark~\ref{rem:connection-quantum-schaefer}.

Complexity bounds for deciding if a quantifier-free first-order formula is satisfied in a partial structure of a lattice-theoretic algebraic theory are established in~\cite{complexityLattices}.
In particular, it is shown that for the theory of Boolean algebras, this decision problem is $\NP$-complete.
However, the notion of partial Boolean algebra employed in this paper is different from the Kochen-Specker notion of partial Boolean algebra we employ in this paper.
In particular, partial Boolean algebras in the sense of~\cite{complexityLattices} are not equipped with a commeasurability relation and are embeddable into a total Boolean algebra.
In light of the Kochen-Specker theorem, these structures are uninteresting for the purpose of understanding contextuality.

\section{Preliminaries}
\label{sec:prelim}
We use $\mathbb{N}$, $\mathbb{R}$, and $\mathbb{C}$ to denote the set of natural, real and complex numbers, respectively.
We use the notation $\FK$ to denote either $\mathbb{R}$ or $\mathbb{C}$ as a field.
For every $n \in \mathbb{N}$, $\setn = \{1,\dots,n\}$.
An $n$-tuple is a function $\vec{a}\colon \setn \rightarrow A$ whose range is denoted $[\vec{a}]$.
We enumerate $n$-tuples through the notation $(a_1,\dots,a_n)$ where $a_i = \vec{a}(i)$.
Let $A^I$ denote the set of functions of type $I \rightarrow A$.
For $n \in \mathbb{N}$, we use the notation $A^n$ when $I = \setn$ and for $n$-tuples $\vec{a}$, we use the notation $A^{\vec{a}}$ for $I = [\vec{a}]$.
Given a function $f\colon A \rightarrow B$ and a subset $D \subseteq A^I$, we use $\res{f}{D}\colon D \rightarrow B^I$ to denote the function defined as $\res{f}{D}(z) = \lambda i \in I. f(z(i))$ for all $z \in D$.

We assume familiarity with the standard syntax and truth table semantics of propositional logic.
We also assume familiarity with the axioms of Boolean algebras.
Propositional formulas are formed from closing a countably infinite set of propositional variables $\Var$ and constants $\{\top\}$ under the unary symbol $\neg$ and binary symbols $\vee$,$\wedge$.
We use the symbols $\impl$, $\cimpl$, $\bimpl$, and $\xor$ as abbreviations for implication, converse implication, bi-implication, and exclusive-or.
Let $\Var(\phi)$ denote the subset of propositional variables which appear in $\phi$.
The notation $\phi(p_1,\dots,p_n)$ indicates that $\Var(\phi) \subseteq \{p_1,\dots,p_n\}$.
Let $\Sub(\phi)$ denote the set of subformulas of $\phi$.
A formula which is either a variable, a constant, or a negation $\neg$ applied to a variable or constant is a \emph{literal}.
We also assume familiarity with standard syntax and model-theoretic semantics of first-order logic.
To distinguish from propositional formulas, first-order formulas and their variables are denoted using captialised letters, i.e.\ $\Phi(P_1,\dots,P_n)$.
\begin{definition}
A \emph{partial Boolean algebra (pBA)} $A$ is a set equipped with
\begin{itemize}
    \item distinguished elements $0_A,1_A \in A$,
    \item a reflexive, symmetric binary relation $\odot_A \subseteq A \times A$,
    \item a function $\neg_A\colon A \rightarrow A$,
    \item and functions $\vee_A,\wedge_A\colon \odot_A \rightarrow A$
\end{itemize}
satisfying the following condition:
\begin{description}
  \item[\namedlabel{cond:ext}{(Ext)}(Ext)] For every set $S \subseteq A$ of pairwise $\odot$-related elements, there exists a $T \supseteq S \cup \{1_A,0_A\}$ of pairwise $\odot$-related elements such that $(T,0_A,1_A,\res{\neg}{T},\res{\vee}{T},\res{\wedge}{T})$ is a Boolean algebra.
\end{description}
The relation $\odot_A$ is called the \emph{commeasurability relation} and we say that \emph{$a \in A$ and $b \in A$ are commeasurable} if $a \odot_{A} b$.
\end{definition}
A partial Boolean algebra $A$ is \emph{total} if $\odot_A = A^2$.
We sometimes drop the adjective `total' when referring to total Boolean algebras.
The trivial Boolean algebra, denoted $\One$, is the unique Boolean algebra where $1_A = 0_A$.
We use $\Total$ to denote the class of all total non-trivial Boolean algebras.
We use $\Two$ to denote the two-element total Boolean algebra.

A set function $f\colon A \rightarrow B$ between partial Boolean algebras is a \emph{homomorphism or $\pBA$ morphism} if $f$ preserves the interpretation of commeasurability relations, the distinguished elements, and the operations $\neg,\vee,\wedge$ whenever defined.
We use $\pBA$ to denote both the category of partial Boolean algebras and partial Boolean algebra homomorphisms, and the class of all partial Boolean algebras.
We use $\BA$ to denote the category of total Boolean algebras and Boolean algebra homomorphisms.

\begin{definition}
\label{def:extension}
Given a partial Boolean algebra $A$ and relation $\oext \subseteq A \times A$, we can construct the \emph{$\oext$-extension of $A$}, denoted $A[\oext]$, which satisfies the following two properties:
\begin{enumerate}[label=(E\arabic*)]
    \item \label{item:extension-existence} There exists a $\pBA$ morphism $\eta\colon A \rightarrow A[\circledcirc]$ satisfying $a \oext b \Rightarrow \eta(a) \odot_{A[\oext]} \eta(b)$.
    \item \label{item:extension-universal} For every $\pBA$ morphism $h\colon A \rightarrow B$ satisfying $a \oext b \Rightarrow h(a) \odot_{B} h(b)$, there exists a unique morphism $\hat{h}\colon A[\circledcirc] \rightarrow B$ such that $\hat{h} \circ \eta = h$.
\end{enumerate}
\end{definition}
Intuitively, $A[\oext]$ forces elements in $A$ which are $\oext$-related to be commeasurable.
Thus, $A[\oext]$ freely adds to $A$ new terms which witness the output of the operations $\vee_{A[\oext]}$ and $\wedge_{A[\oext]}$.
The explicit construction of $A[\oext]$ is given in Section 2.2 of~\cite{logicContextuality}.
We make use of the properties \ref{item:extension-existence}-\ref{item:extension-universal} which $A[\oext]$ satisfies by Theorem 1 and Proposition 3 of~\cite{logicContextuality}.

The category $\pBA$ has general colimits.
This fact was first proved in~\cite{noncommutativityAsColimit}, but the proof appealed to the Adjoint Functor Theorem and thus did not yield an explicit construction.  An explicit construction of colimits in terms of coproducts and a quotient construction satisfying a universal property similar to the $\oext$-extension of $A$ were given in~\cite{logicContextuality}

\begin{definition}
\label{def:quotient}
Given a partial Boolean algebra $A$ and binary relation $\oext$, the \emph{$\oext$-quotient of $A$}, denoted $A\qs \oext$ is a partial Boolean algebra satisfying the following two properties:
\begin{enumerate}[label=(Q\arabic*)]
    \item \label{item:quotient-existence} There exists a $\pBA$ morphism $\upsilon\colon A \rightarrow A \qs \oext$ satisfying $a \circledcirc b \Rightarrow \upsilon(a) = \upsilon(b)$.
    \item \label{item:quotient-universal} For every $\pBA$ morphism $h\colon A \rightarrow B$ satisfying $a \oext b \Rightarrow h(a) = h(b)$, there exists a unique morphism $\bar{h}\colon A \backslash \oext \rightarrow B$ such that $\bar{h} \circ \upsilon = h$.
\end{enumerate}
\end{definition}
The explicit construction is a simple modification of the construction for $A[\oext]$ and is also detailed in Section 2.2 of~\cite{logicContextuality}.
As with the extension construction $A[\oext]$, we only make use of the properties \ref{item:quotient-existence}-\ref{item:quotient-universal} which $A \qs \oext$ satisfies by Theorem 5 of~\cite{logicContextuality}.

Given a formula $\phi(p_1,\dots,p_n)$ and Boolean algebra $B$, there is a substitution function $\phi^{B}\colon B^{\vec{p}} \rightarrow B$, which gives the value of the formula, given any valuation of the variables.
In this context of partial Boolean algebras, the notion of substitution function must be adapted to meaningful substitution.
\begin{definition}
\label{def:substitution}
Given a propositional formula $\phi(\vec{p})$ with variables among $\vec{p} = (p_1,\dots,p_n)$ and partial Boolean algebra $A$, we define a \emph{meaningful domain} $\mdomA{\phi(\vec{p})} \subseteq A^{\vec{p}}$ and \emph{meaningful substitution function $\phi^{A}\colon \mdomA{\phi(\vec{p})} \rightarrow A$} mutually by structural recursion on propositional formulas $\phi(p_1,\dots,p_n)$.
  For the base case,
  \begin{enumerate}[label=(\arabic*)]
    \item If $\phi(\vec{p}) = \top$, then $\mdomA{\phi(\vec{p})} = A^{\vec{p}}$ and for all $\alpha \in \mdomA{\phi(\vec{p})}$, $\phi^A(\alpha) = 1_A$.
    \item If $\phi(\vec{p}) = p_i$, then $\mdomA{\phi(\vec{p})} = A^{\vec{p}}$ and for all $\alpha \in \mdomA{\phi(\vec{p})}$, $\phi^A(\alpha) = \alpha(p_i)$.
  \end{enumerate}
  For the inductive steps,
  \begin{enumerate}[label=(\arabic*), resume]
    \item If $\phi(\vec{p}) = \neg \psi(\vec{p})$, then $\mdomA{\phi(\vec{p})} = \mdomA{\psi(\vec{p})}$ and for all $\alpha \in \mdomA{\phi(\vec{p})}$, $\phi^{A}(\alpha) = \neg_{A} \psi^{A}(\alpha)$.
    \item \label{item:def-substitution-connective} If $\phi(\vec{p}) = \psi_1(\vec{p}) \bowtie \psi_{2}(\vec{p})$ for connective $\bowtie \in \{\vee,\wedge\}$, then
      \[ \mdomA{\phi(\vec{p})} = \{\alpha \in \mdomA{\psi_1(\vec{p})} \cap \mdomA{\psi_2(\vec{p})} \mid \psi_1^A(\alpha) \odot_A \psi_2^A(\alpha) \} \]
      and for all $\alpha \in \mdomA{\phi(\vec{p})}$, $\phi^{A}(\alpha) = \psi_1^{A}(\alpha) \bowtie_A \psi_2^{A}(\alpha)$.
  \end{enumerate}
\end{definition}
For cleaner notation, for a propositional formula $\phi(p_1,\dots,p_n)$, we sometimes use $\phi^{A}(a_1,\dots,a_n)$ to denote $\phi^A(\alpha)$ where $\alpha \in \mdomA{\phi(\vec{p})}$ and for every $i \in \setn$, $\alpha(p_i) = a_i$.

In the next section, we use the notion of meaningful substitution to generalise propositional satisfiability and validity to any class of partial Boolean algebras.

In ordinary universal algebra, and in particular Boolean algebras, satisfiability of a formula is preserved and reflected by $\BA$ morphisms.
The following proposition from~\cite{hiddenVariableQM} extends this result to meaningful substitutions and $\pBA$ morphisms.
\begin{proposition}
\label{prop:morphism-and-substitution}
  If $h\colon A \rightarrow B$ is a $\pBA$ morphism and $\phi(\vec{p})$ is propositional formula, then for all $\alpha \in \mdomA{\phi(\vec{p})}$,
  \[ \phi^{B}(\res{h}{\mdomA{\phi(\vec{p})}}(\alpha)) = h(\phi^{A}(\alpha)). \]
\end{proposition}

We also make use of Proposition~\ref{prop:equality-bi-implication} and Proposition~\ref{prop:order-expressible} which generalise useful facts about Boolean algebras to partial Boolean algebras.
\begin{proposition}
  \label{prop:equality-bi-implication}
  Let $A$ be a partial Boolean algebra, $\phi(\vec{p}),\psi(\vec{p})$ be propositional formulas and $\alpha \in \mdomA{\phi(\vec{p})} \cap \mdomA{\psi(\vec{p})}$.
  \[ \phi^A(\alpha) = \psi^A(\alpha) \text{ if, and only if, } \nu^{A}(\alpha) = 1_A \]
  where $\nu(\vec{p})$ is the bi-implication $\phi(\vec{p}) \bimpl \psi(\vec{p})$.
\end{proposition}
\begin{proof}
The proof of this statement is the same as the $\Total$ case, but we have the additional burden of checking commeasurablity.
Throughout this proof, let $a_{\phi} = \phi^A(\alpha)$ and $a_{\psi} = \psi^A(\alpha)$.

For the $\Leftarrow$ direction, by hypothesis, $a_{\phi} = a_{\psi}$.
Since $\odot_A$ is reflexive, we have that $a_{\phi} \odot_A a_{\psi}$.
Therefore, by \ref{cond:ext}, $a_{\phi}$ and $a_{\psi}$ are contained within a total Boolean subalgebra $B$ of $A$, and we can proceed with the standard Boolean algebra equational proof to demonstrate that $a_{\phi} \bimpl_A a_{\psi} = 1_A$.
For completeness, we reproduce that proof here:
\begin{align*}
  a_{\phi} \leftrightarrow_A a_{\psi} &= (a_{\phi} \wedge_A a_{\psi}) \vee_A (\neg a_{\phi} \wedge_A \neg a_{\psi}) \\
  &= (a_{\phi} \wedge_A a_{\phi}) \vee_A (\neg a_{\phi} \wedge_A \neg a_{\phi}) \\
  &= a_{\phi} \vee_A \neg a_{\phi} \\
  &= 1_A
\end{align*}
For the $\Rightarrow$ direction, we first give a standard equational proof that $a_{\phi} = a_{\phi} \wedge_A a_{\psi}$ if $a_{\phi} \leftrightarrow a_{\psi} = 1_A$.
\begin{align*}
  a_{\phi} &= a_{\phi} \wedge_A 1_A \\
           &= a_{\phi} \wedge_A (a_{\phi} \leftrightarrow_A a_{\psi}) \\
           &= a_{\phi} \wedge_A ((a_{\phi} \wedge_A a_{\psi}) \vee_A (\neg a_{\phi} \wedge_A \neg a_{\psi})) \\
           &= (a_{\phi} \wedge_A (a_{\phi} \wedge_A a_{\psi})) \vee_A (a_{\phi} \wedge_A (\neg a_{\phi} \wedge_A \neg a_{\psi})) \\
           &= ((a_{\phi} \wedge_A a_{\phi}) \wedge_A a_{\psi}) \vee_A ((a_{\phi} \wedge_A \neg a_{\phi}) \wedge_A \neg a_{\psi}) \\
           &= (a_{\phi} \wedge_A a_{\psi}) \vee_A ((a_{\phi} \wedge_A \neg a_{\phi}) \wedge_A \neg a_{\psi}) \\
           &= (a_{\phi} \wedge_A a_{\psi}) \vee_A (0_A \wedge_A \neg a_{\psi}) \\
           &= (a_{\phi} \wedge_A a_{\psi}) \vee_A 0_A \\
           &= a_{\phi} \wedge_A a_{\psi}
\end{align*}
In this proof, we used the fact that $a_{\phi} \odot_A 1_A$ which follows by \ref{cond:ext}, and $a_{\phi} \odot_A a_{\psi}$ which follows by the hypothesis that $\nu^{A}(\alpha) = 1_A$, $\phi(\vec{p}) \wedge \psi(\vec{p})$ being a subformula of $\nu(\vec{p})$, and case~\ref{item:def-substitution-connective} of Definition~\ref{def:substitution}.
A similar proof demonstrates that $a_{\psi} = a_{\phi} \wedge_A a_{\psi}$.
Therefore, $a_{\phi} = a_{\phi} \wedge_A a_{\psi} = a_{\psi}$ in $A$ as desired.
\end{proof}

For every partial Boolean algebra $A$, we define the relation $\leq_{A}$ as $a \leq_A b$ if $a \odot_A b$ and $a \wedge_A b = a$.
For every Boolean subalgebra $B$ of $A$, $\leq_A$ restricted to $B$ coincides with the underlying partial order on the Boolean algebra $B$.
The relation $\leq_{A}$ on a partial Boolean algebra is expressible in propositional logic, since we can define the formula
\begin{equation}
  \label{eq:order-expressible}
  \phi_{\leq}(p,q) := p \wedge q \leftrightarrow p
\end{equation}
which satisfies the following proposition.
\begin{proposition}
  \label{prop:order-expressible}
  Let $A$ be a partial Boolean algebra with $a,b \in A$.
  \[ a \leq_{A} b \text{ if, and only if, } \phi_{\leq}^A(a,b) = 1_A \]
\end{proposition}
\begin{proof}
  Apply Proposition~\ref{prop:equality-bi-implication} where $\phi(p,q) = p \wedge q$ and ${\psi(p,q) = p}$.
\end{proof}

\section{Formulation and observations}
\label{sec:observations}
We start with the formulation of the primary motivation for our paper: satisfiability in classes of partial Boolean algebras.
Using Definition~\ref{def:substitution}, we obtain two sensible notions of satisfiability.
Given a partial Boolean algebra $A$ and propositional formula $\phi(\vec{p})$, we say
\begin{enumerate}
  \item \emph{$\phi$ is weakly satisfied in $A$} if there exists a $\alpha \in \mdomA{\phi(\vec{p})}$ such that $\phi^{A}(\alpha) \not= 0_A$.
  \item \emph{$\phi$ is strongly satisfied in $A$} if there exists a $\alpha \in \mdomA{\phi(\vec{p})}$ such that $\phi^{A}(\alpha) = 1_A$.
\end{enumerate}
Similarly, if we are given a class of partial Boolean algebras $\Cls$,
\begin{enumerate}
  \item \emph{$\phi$ is $\Cls$-weakly-satisfiable} if $\phi$ is weakly satisfied in $A$ for some $A \in \Cls$.
  \item \emph{$\phi$ is $\Cls$-strongly-satisfiable} if $\phi$ is strongly satisfied in $A$ for some $A \in \Cls$.
\end{enumerate}
We can also generalise the dual notion of validity to any class of partial Boolean algebras.
We say that $\phi$ is $\Cls$-weakly-valid or $\Cls$-strongly-valid if for all $A \in \Cls$ and $\alpha \in \mdomA{\phi(\vec{p})}$, $\phi(\alpha) \not= 0_A$ or $\phi(\alpha) = 1_A$, respectively.
Note that if $A,B$ are isomorphic partial Boolean algebras, then $A$ and $B$ satisfy the same class of formulas.
Thus, for the notions of satisfiability and validity we are studying, we assume without loss of generality throughout the rest of the paper that every class of partial Boolean algebras $\Cls$ is isomorphism-closed.

\begin{proposition}
Let $\Cls,\Clsp$ be classes of partial Boolean algebras such that $\Cls \subseteq \Clsp$ and $\phi$ is a propositional formula.
\begin{enumerate}
  \item If $\phi$ is $\Cls$-weakly-satisfiable, then $\phi$ is $\Clsp$-weakly-satisfiable.
  \item If $\phi$ is $\Clsp$-weakly-valid, then $\phi$ is $\Cls$-weakly-valid.
\end{enumerate}
Similar statements hold for the corresponding strong notions.
\end{proposition}
For some classes $\Cls$ of partial Boolean algebras, the collection of formulas $\phi$ which are $\Cls$-weakly-satisfiable and $\Cls$-strongly-satisfiable coincide.
In these cases (which are all cases of interest in the subsequent sections), we can drop the qualifiers and say that a formula $\phi$ is $\Cls$-satisfiable or dually, is $\Cls$-valid.
It is a well-known fact that $\Total$ is a class where $\Total$-weakly-satisfiable and $\Total$-strongly-satisfiable coincide.
This is because $\Total$ is closed under taking a quotient $B \qs \langle b \rangle$ of a Boolean algebra $B \in \Total$ by the filter generated from a non-zero element $b \in B$.
The quotient $B \qs \langle b \rangle$ `forces' $b \in B$ to be equal to one.
This idea generalises to an arbitrary class $\Cls$ which is closed under taking a $\oext$-quotient $A \qs \oext$ of $A \in \Cls$ by $\oext = \{(a,1_A)\}$ for a non-zero element $a \in A$.
The following definition isolates such classes $\Cls$.
\begin{definition}
  \label{def:closed-under-collapse}
  A class $\Cls$ is \emph{closed under collapse} if for all $A \in \Cls$ and $a \in A$ such that $a \not= 0_A$, $A \qs \oext \in \Cls$ where $\oext = \{(a,1_A)\}$.
\end{definition}

\begin{proposition}
\label{prop:strong-weak}
If $\Cls$ is closed under collapse and $\One \not\in \Cls$, then for any propositional formula $\phi$, $\phi$ is $\Cls$-weakly-satisfiable if, and only if, $\phi$ is $\Cls$-strongly-satisfiable.
\end{proposition}
\begin{proof}
  $\Rightarrow$ If $\phi$ is $\Cls$-weakly-satisfiable, then there exists $\alpha \in \mdomA{\phi(\vec{p})}$ such that $\phi^A(\alpha) \not= 0_A$.
  Since $\Cls$ is closed under collapse, $A \qs \oext \in \Cls$ where $\oext$ identifies $\phi^A(\alpha)$ with $1_A$.
  By~\ref{item:quotient-existence}, there exists a $\pBA$-morphism $\upsilon\colon A \rightarrow A \qs \oext$ such that $\upsilon(\phi^A(\alpha)) = \upsilon(1_A)$.
  Thus, by $\upsilon$ being a $\pBA$-morphism, $\phi^{A \qs \oext}(\res{\upsilon}{\mdomA{\phi(\vec{p})}}(\alpha)) = \upsilon(\phi^A(\alpha)) = \upsilon(1_A) = 1_{A \qs \oext}$.

$\Leftarrow$ Conversely, since $\One \not\in \Cls$, for every $A \in \Cls$, $1_A \not= 0_A$.
\end{proof}

From the notions of $\Cls$-weakly-satisfiable and $\Cls$-strongly-satisfiable, we can formulate the corresponding decision problems.
$\weakSAT{\Cls}$ and $\strongSAT{\Cls}$ denote the classes of $\Cls$-weakly-satisfiable or $\Cls$-strongly-satisfiable formulas, respectively.
For classes $\Cls$ where these notions coincide, we use the notation $\SAT{\Cls}$ for this decision problem.
In the case where $\Cls$ is the singleton class $\{A\}$, we drop the superfluous braces when denoting these decision problems, i.e.\ $\strongSAT{A}$.

Since the trivial Boolean algebra $\One$ satisfies every term, we observe that if $\One \in \Cls$, then $\strongSAT{\Cls}$ is trivial, i.e.\ it contains all formulas.
By contrast, suppose $\Cls$ is such that $\One \not\in \Cls$, then since any propositional formula which is $\Total$-satisfied is satisfied in $\Two$ and any non-trivial partial Boolean algebra contains $\Two$ as a Boolean subalgebra, $\weakSAT{\Cls}$ and $\strongSAT{\Cls}$ should be at least as computationally hard as $\SAT{\Total}$.
Corollary~\ref{cor:csat-np-hard} of Proposition~\ref{prop:sat-reduction} below confirms this intuition.

\begin{proposition}
\label{prop:sat-reduction}
Let $\Cls$ be a class of partial Boolean algebras such that $\One \not\in \Cls$.
For every propositional formula $\phi(\vec{p})$, there exists a propositional formula $\hat{\phi}(\vec{p})$ of length $O(|\phi(\vec{p})|^2)$ such that the following are equivalent:
\begin{enumerate}[label=(\arabic*)]
  \item \label{item:sat-reduction-total} $\phi(\vec{p}) \in \SAT{\Total}$.
  \item \label{item:sat-reduction-strong} $\hat{\phi}(\vec{p}) \in \strongSAT{\Cls}$.
  \item \label{item:sat-reduction-weak} $\hat{\phi}(\vec{p}) \in \weakSAT{\Cls}$.
\end{enumerate}
\end{proposition}
\begin{proof}
Suppose $\vec{p} = (p_1,\dots,p_n)$, then we use a $\Total$-valid $\psi(\vec{p})$ formula defined as:
\[\psi(\vec{p}) := \bigwedge_{i \not= j \in \setn} (p_i \wedge p_j) \vee (\neg p_i \wedge p_j) \vee (p_i \wedge \neg p_j) \vee (\neg p_i \wedge \neg p_j) \]
to define $\hat{\phi}(\vec{p})$ as the conjunction:
\[ \hat{\phi}(\vec{p}) := \phi(\vec{p}) \wedge \psi(\vec{p}). \]
Intuitively, each subformula $p_i \wedge p_j$ in $\psi(\vec{p})$ forces the elements assigned to the variables $p_i,p_j$ in a partial Boolean algebra to be commeasurable.
In order to ensure the addition of this subformula has no effect on satisfiablity, we take the disjunction of the subformula $p_i \wedge p_j$ with the variants $\neg p_i \wedge p_j$, $p_i \wedge \neg p_j$, $\neg p_i \wedge \neg p_j$ which forces the disjunction of all these variants to be $\Total$-valid, and thus $\psi(\vec{p})$ is $\Total$-valid.
We now verify that $\hat{\phi}(\vec{p})$ satisfies the desired property.

$\ref{item:sat-reduction-total} \Rightarrow \ref{item:sat-reduction-strong}$ Suppose $\phi(\vec{p})$ is $\Total$-satisfiable.
Since $\psi(\vec{p})$ in $\hat{\phi}(\vec{p})$ is $\Total$-valid, $\hat{\phi}(\vec{p})$ is $\Total$-satisfiable.
Since every $B \in \Total$ has an homomorphism $h\colon B \rightarrow \Two$, $\Two$ strongly satisfies $\hat{\phi}(\vec{p})$.
Now, consider $A \in \Cls$, by hypothesis, $A$ is non-trival, so $1_{A} \not= 0_{A}$.
Thus, by \ref{cond:ext}, $\Two$ is a Boolean subalgebra of $A$, so there exists an inclusion $\Two \hookrightarrow A$.
Therefore, by Proposition~\ref{prop:morphism-and-substitution}, $A$ strongly satisfies $\hat{\phi}(\vec{p})$ and $\hat{\phi}(\vec{p})$ is $\Cls$-strongly-satisfiable.

$\ref{item:sat-reduction-strong} \Rightarrow \ref{item:sat-reduction-weak}$ Since $\One \not\in \Cls$, for every $A \in \Cls$, $1_A \not= 0_A$.

$\ref{item:sat-reduction-weak} \Rightarrow \ref{item:sat-reduction-total}$
Suppose $\hat{\phi}(\vec{p})$ is $\Cls$-weakly-satisfiable.
By the definition of $\Cls$-weakly-satisfiable and $\One \not\in \Cls$, there exists a non-trivial partial Boolean algebra $A \in \Cls$ and $\alpha \in \mdomA{\phi}$ such that $\hat{\phi}^{A}(\alpha) \not= 0_A$.
Since for every $i \not= j \in \setn$, $p_i \wedge p_j$ is a subformula of $\hat{\phi}(\vec{p})$, we can conclude from Definition~\ref{def:substitution} that $\alpha(p_i) \odot_A \alpha(p_j)$.
Moreover, as $\odot_A$ is reflexive, it follows that the subset $\{\alpha(p_1),\dots,\alpha(p_n)\}$ is pairwise $\odot_A$-related.
By axiom \ref{cond:ext} of partial Boolean algebras, there exists a total Boolean subalgebra $B$ of $A$ such that $\{\alpha(p_1),\dots,\alpha(p_n)\} \subseteq B$.
Since $B$ is a total Boolean subalgebra $\hat{\phi}^{B}(\alpha) \in B$.
By $\hat{\phi}^{A}(\vec{a}) \not= 0_A$, we have that $\hat{\phi}^{B}(\alpha) \not= 0_B$ and $B$ is non-trivial, so $B \in \Total$.
Therefore, $\hat{\phi}(\vec{p})$ and its conjunct $\phi(\vec{p})$ is $\Total$-satisfiable.
\end{proof}

\begin{corollary}
  \label{cor:csat-np-hard}
  If $\Cls$ is a class of partial Boolean algebras such that $\One \not\in \Cls$, then $\weakSAT{\Cls}$ and $\strongSAT{\Cls}$ is $\NP$-hard.
\end{corollary}
\begin{proof}
  Proposition~\ref{prop:sat-reduction} gives a polynomial time reduction from $\SAT{\Total}$ to $\weakSAT{\Cls}$ and $\strongSAT{\Cls}$.
  By the Cook-Levin theorem, $\SAT{\Total}$ is $\NP$-hard.
\end{proof}
Many of the classes $\Cls$ of partial Boolean algebras we consider satisfy Corollary~\ref{cor:csat-np-hard}, and so $\SAT{\Cls}$ is $\NP$-hard.
However, it remains to establish upper bounds and, where relevant, tighter lower bounds for the complexity of $\SAT{\Cls}$.

One strengthening of the Cook-Levin theorem, discovered by Tseitin in~\cite{cnfTransform}, showed that the $\SAT{\Total}$ problem remains $\NP$-complete when restricted to formulas $\phi(\vec{p})$ in $3$-literal-per-clause conjunctive normal form (3CNF), i.e.\ $\phi(\vec{p})$ is a conjunction of disjunctive clauses each consisting of only at most 3 literals.
Let $\strongCNF{\Cls}$ and $\weakCNF{\Cls}$ denote the $\Cls$-strong-satisfiability and $\Cls$-weak-satisfiablity decision problems where the inputs are restricted to CNF formulas.

The key insight of Tseitin was that every propositional formula $\phi(\vec{p})$ could be transformed, in polynomial time, into a $\Total$-equisatisfiable, but not necessarily equivalent, CNF formula $\phi_{\CNF}(\vec{p},\vec{q})$.
Proposition~\ref{prop:cnf-pba} shows that the Tseitin transformation preserves $\Cls$-strong-satisfiability for any class $\Cls$ of partial Boolean algebras.
The Tseitin transformation on a propositional formula $\phi(\vec{p})$ produces a CNF formula $\phi_{\CNF}(\vec{p},\vec{q})$ where $\vec{q}$ is such that
\[ [\vec{q}] = \{q_{\psi} \mid \psi \text{ is non-variable subformula of $\phi$ }\}. \]
We define $\phi_{\CNF}$ and a set of clauses $C_{\phi}$ by induction on the structure of $\phi(\vec{p})$.
For the base cases,
\begin{enumerate}
  \item For $\phi(\vec{p}) = \top$, $\phi_{\CNF}(\vec{p},\vec{q}) = \top$ and $C_{\phi} = \varnothing$.
  \item For $\phi(\vec{p}) = p_i$, $\phi_{\CNF}(\vec{p},\vec{q}) = p_i$ and $C_{\phi} = \varnothing$.
\end{enumerate}
For the inductive steps, we define $\phi_{\CNF}(\vec{p},\vec{q})$ as $q_{\phi} \wedge \bigwedge C_{\phi}$ where $C_{\phi}$ depends on the inductive case.
In the following, we assume $v_{p_i} = p_i$ and for non-variable subformulas $\psi$ of $\phi$, $v_{\psi} = q_{\psi}$.
\begin{enumerate}
  \item \label{item:tseitin-negation} For $\phi(\vec{p}) = \neg \psi(\vec{p})$, let $C_{\phi} = C_{\psi} \cup N$ where
    \[  N = \{\neg v_{\phi} \vee \neg v_{\psi},v_{\phi} \vee v_{\psi}\}. \]
    The conjunction of $N$ is equivalent to $v_{\phi} \bimpl \neg v_{\psi}$.
  \item \label{item:tseitin-disjunction} For $\phi = \psi_1 \vee \psi_2$, let $C_{\phi} = C_{\psi_1} \cup C_{\psi_2} \cup O$ where
    \[ O = \{\neg v_{\phi} \vee v_{\psi_1} \vee v_{\psi_2}, v_{\phi} \vee \neg v_{\psi_1}, v_{\phi} \vee \neg v_{\psi_2} \}. \]
    The conjunction of $O$ is equivalent to $v_{\phi} \bimpl v_{\psi_1} \vee v_{\psi_2}$.
  \item \label{item:tseitin-conjunction} For $\phi = \psi_1 \wedge \psi_2$, let $C_{\phi} = C_{\psi_1} \cup C_{\psi_2} \cup W$ where
    \[ W = \{v_{\phi} \vee \neg v_{\psi_1} \vee \neg v_{\psi_2},\neg v_{\phi} \vee v_{\psi_1},\neg v_{\phi} \vee v_{\psi_2} \}. \]
    The conjunction of $W$ is equivalent to $v_{\phi} \bimpl v_{\psi_1} \wedge v_{\psi_2}$.
\end{enumerate}

\begin{proposition}
\label{prop:cnf-pba}
  Let $A$ be a partial Boolean algebra and $\phi(\vec{p})$ a propositional formula.
  \begin{center}
    $\phi$ is strongly satisfied in $A$ iff $\phi_{\CNF}$ is strongly satisfied in $A$
  \end{center}
\end{proposition}
\begin{proof}
The proof for this statement is similar to the $\Total$ case.
The key insight that allows us to generalise to arbitrary partial Boolean algebra is that for subformulas of $\phi$ with the form $\psi = \psi_1 \bowtie \psi_2$ where $\bowtie \in \{\vee,\wedge\}$ and $\alpha \in \mdomA{\phi(\vec{p})}$, $\{\psi^A(\alpha),\psi^A_1(\alpha),\psi^A_2(\alpha)\}$ is pairwise commeasurable.
Thus, we can construct a $\beta \in A^{\vec{q}\vec{p}}$ such that $\beta \in \mdomA{\phi_{\CNF}}$ is a meaningful substitution.

For the $\Rightarrow$ direction, suppose $\phi(\vec{p})$ is strongly satisfied in $A$.
We proceed by structural induction on $\phi$ to demonstrate that the conjunction $\bigwedge C_{\phi}$ in $\phi_{\CNF}(\vec{p},\vec{q})$ is strongly satisfied in $A$.
For the base cases $\phi(\vec{p}) = \top$ and $\phi(\vec{p}) = p_i$, $C_{\phi} = \varnothing$.
The empty conjunction $\bigwedge C_{\phi}$ is interpreted as $1_{A}$ for every pBA $A$.
For the inductive step, we first spell out the case of $\phi(\vec{p}) = \psi_1(\vec{p}) \wedge \psi_2(\vec{p})$.
By Proposition~\ref{prop:equality-bi-implication}, we can conclude that $\phi(\vec{p}) \leftrightarrow \psi_1(\vec{p}) \wedge \psi_2(\vec{p})$ is strongly satisfied via the meaningful substitution $\alpha \in \mdomA{\phi(\vec{p})}$.
Consider the substitution $\beta \in A^{\vec{p}\vec{q}}$ such that $\beta(q_{\phi}) = \phi^{A}(\alpha)$, $\beta(q_{\psi_i}) = \psi_i^{A}(\alpha)$ for $i \in \{1,2\}$, and $\res{\beta}{[\vec{p}]} = \alpha$.
Observe that $\beta \in \mdomA{\phi_{\CNF}}$ is a meaningful substitution since $\{\phi^{A}(\alpha), \psi_1^{A}(\alpha), \psi_2^{A}(\alpha)\}$ is a pairwise commeasurable set.
Under the substitution $\beta$, since $\phi(\vec{p}) \leftrightarrow \psi_1(\vec{p}) \wedge \psi_2(\vec{p})$ is strongly satisfied via $\alpha$, the clauses $W$ in $C_{\phi}$ expressing $q_{\phi} \leftrightarrow q_{\psi_1} \wedge q_{\psi_2}$ are strongly satisfied via $\beta$.
By the inductive hypothesis, we can conclude the conjunctions $\bigwedge C_{\psi_i}$ for $i \in \{1,2\}$, and thus the clauses $C_{\psi_1} \cup C_{\psi_2} \subseteq C_{\phi}$ are also strongly satisfied in $A$.
Therefore, the entire conjunction $\bigwedge C_{\phi}$ is strongly satisfied in $A$.
The proofs for the other inductive cases, $\neg$ and $\vee$, are similar.
Finally, by hypothesis $\phi(\vec{p})$ is strongly satisfied in $A$, so $\phi_{\CNF}(\vec{p},\vec{q})$ is also strongly satisfied in $A$.

For the $\Leftarrow$ direction, suppose $\phi_{\CNF}(\vec{p},\vec{q})$ is strongly satisfied.
By definition, there exists a $\beta \in \mdomA{\phi_{\CNF}}$ such that $\phi^A_{\CNF}(\beta) = 1_A$.
By standard Boolean algebra, it follows that for each conjunct $\chi$ of $\phi_{\CNF}$, $\chi^A(\beta) = 1_A$.
By construction of $\phi_{\CNF}$, every conjunct being strongly satisfied amounts to asserting the bi-implications $q_{\psi} \leftrightarrow q_{\psi_1} \bowtie q_{\psi_2}$ for $\bowtie \in \{\vee,\wedge\}$ or $\neg q_{\psi} \leftrightarrow q_{\neg \psi}$.
By inductively applying Proposition~\ref{prop:equality-bi-implication}, we obtain that $\beta(q_{\phi}) = \phi^A(\alpha)$ where $\alpha = \res{\beta}{[\vec{p}]}$.
Since $q_{\phi}$ itself is a conjunct of $\phi_{\CNF}(\vec{p},\vec{q})$, $\beta(q_{\phi}) = \phi^A(\alpha) = 1_A$.
\end{proof}
Thus, we can conclude that for all classes $\Cls$, the problems $\strongCNF{\Cls}$ and $\strongSAT{\Cls}$ are in the same complexity class up to polynomial reductions.

Moreover, by modifying the formula used in the Proposition~\ref{prop:sat-reduction}, we can prove a refinement of Corollary~\ref{cor:csat-np-hard} demonstrating that $\strongCNF{\Cls}$ and $\weakCNF{\Cls}$ are also $\NP$-hard.
\begin{theorem}
  \label{thm:cnf-csat-np-hard}
  If $\Cls$ is a class of partial Boolean algebras such that $1 \not\in \Cls$, then $\strongCNF{\Cls}$ and $\weakCNF{\Cls}$ is $\NP$-hard.
\end{theorem}
\begin{proof}
  In the proof of Proposition~\ref{prop:sat-reduction}, for any $\Total$-satisfiable formula $\phi(\vec{p})$, we defined a $\Cls$-(weak/strong)-satisfiable formula as $\hat{\phi} = \phi(\vec{p}) \wedge \psi(\vec{p})$ where $\psi(\vec{p})$ was a classical tautology in which every pair of distinct variables $p_i,p_j \in [\vec{p}]$ appeared.
  Similarly, we define a classical CNF tautology $\psi_{\CNF}(\vec{p})$ with the same property on distinct variables:
  \[ \psi_{\CNF}(\vec{p}) = \bigwedge_{i \not= j} (\neg p_i \vee p_i \vee p_j) \wedge (\neg p_j \vee p_j \vee p_i). \]
  Thus, starting with a $\Total$-satisfiable CNF formula $\phi(\vec{p})$, we can produce a $\Cls$-(weak/strong)-satisfiable CNF formula $\hat{\phi}(\vec{p}) = \phi(\vec{p}) \wedge \psi_{\CNF}(\vec{p})$.
  The problems $\strongCNF{\Cls}$ and $\weakCNF{\Cls}$ are $\NP$-hard as $\bothCNF{\Total}$ is $\NP$-complete.
\end{proof}

\section{All satisfiablity}
\label{sec:all-sat}
Before we proceed to investigate the complexity of the problems $\strongSAT{\Cls}$ and $\weakSAT{\Cls}$ for various classes $\Cls$ of pBAs, we first consider the decision problem $\VARSAT$.  This consists of the set of tuples
$(A,a,\phi(\vec{p}))$ where $A$ is a finite pBA, $a \in A$, and $\phi(\vec{p})$ is a propositional formula such that there exists an $\alpha \in A^{\phi(\vec{p})}$ with $\phi^{A}(\alpha) = a$.
The following proposition shows that $\VARSAT$ is $\NP$-complete.
\begin{proposition}
  $\VARSAT$ is $\NP$-complete.
\end{proposition}
\begin{proof}
  To show membership in $\NP$, consider the algorithm which guesses an assignment $\alpha \in A^{\vec{p}}$ non-deterministically.
  Verifiying that $\alpha \in A^{\phi(\vec{p})} \subseteq A^{\vec{p}}$ and $\phi^{A}(\alpha) = a$ can be done in polynomial time.
  The assignment $\alpha$ is linear in the size of $\phi(\vec{p})$.
  To show $\VARSAT$ is $\NP$-hard, we can translate any input instance $\varphi(\vec{p})$ of the classical satisfiablity problem $\SAT{\Total}$ to the instance $(\Two, 1 \in \Two,\varphi(\vec{p}))$ of the $\VARSAT$ problem.
\end{proof}
The first natural class to investigate is the class of all partial Boolean algebras $\pBA$.
However, $\pBA$ includes the trivial partial Boolean algebra $\One$ which satisfies every propositional formula.
Therefore, we instead consider the class $\All$, formally defined as
\[ \All = \{A \in \pBA \mid 0_A \not= 1_A \}, \]
of all non-trivial partial Boolean algebras.

Observe that by construction $\One \not\in \All$ and so, $\All$ is closed under collapse (as in Definition~\ref{def:closed-under-collapse}).
Therefore, by Proposition~\ref{prop:strong-weak}, the notions of $\All$-strong-satisfiablity and $\All$-weak-satisfiablity coincide, and we can consider the $\SAT{\All}$ decision problem without any ambiguity.
Our next observation is that if a $\phi(\vec{p})$ is $\All$-satisfiable, then $\phi(\vec{p})$ is satisfiable in a finite $\pBA$.
Namely, we can freely construct a finite $\pBA$ $F_{\phi}$ for every formula $\phi(\vec{p})$ which always has a term corresponding to $\phi(\vec{p})$.
This would then yield an algorithm for $\SAT{\All}$ by deciding $\VARSAT$ on the input pBA $F_{\phi}$.
By utilising extensions (Definition~\ref{def:extension}) and quotients (Definition~\ref{def:quotient}), we construct $F_{\phi}$ by structural induction on $\phi$.
\begin{enumerate}[label=(\arabic*)]
    \item If $\phi = \top$, then $F_{\phi} = \Two$.
    \item If $\phi = p_i$, then $F_{\phi}$ is the $4$-element Boolean algebra on $\{0,p_i, \neg p_i,1\}$
    \item If $\phi = \neg \psi$, then $F_{\phi} = F_{\psi}$. Note that $[\phi] = \neg [\psi] \in F_{\phi}$
    \item If $\phi = \psi_1 \bowtie \psi_2$ where $\bowtie \in \{\vee,\wedge\}$, then we break down the construction into three steps:
    \begin{itemize}
      \item $H_{\phi} = F_{\psi_1} \uplus F_{\psi_2}$ is the coproduct in $\pBA$ with coprojections $\iota_j\colon F_{\psi_j} \rightarrow F_{\psi_1} \uplus F_{\psi_2}$ for $j \in \{1,2\}$.
      \item $G_{\phi} = H_{\phi}[\circledcirc]$ where $\circledcirc = \{(\iota_{1}([\psi_1]),\iota_2([\psi_2]))\}$.
        By the construction of $\circledcirc$-extensions, this means that there exists a term
        \[ [\phi]_G = [\iota_1([\psi_1]) \bowtie \iota_2([\psi_2])] \in (F_{\psi_1} \uplus F_{\psi_2})[\circledcirc] \]
        By property~\ref{item:extension-existence} of $\circledcirc$-extensions, there exists a morphism $\eta_{\phi}\colon H_{\phi} \rightarrow G_{\phi}$.
      \item $F_{\phi} = G_{\phi} \backslash \circleddash$ where
          \[ \circleddash = \{(\iota_1([\gamma]),\iota_2([\gamma])) \mid \gamma \in \Sub(\psi_1) \cap \Sub(\psi_2)\}. \]
          There is a term $[\phi]_F \in F_{\phi}$ which is the equivalence class with the representative $[\phi]_G \in G_{\phi}$.
        By property~\ref{item:quotient-existence} of $\circleddash$-quotient, there exists a morphism $\nu_{\psi}\colon G_{\phi} \rightarrow F_{\phi}$.
    \end{itemize}
    Explicitly, $F_{\phi} = (F_{\psi_1} \uplus F_{\psi_2})[\circledcirc] \backslash \circleddash$.
\end{enumerate}
Intuitively, $F_{\phi}$ can be viewed as a $\pBA$ analogue to the free Boolean algebra construction in $\BA$.
This intuition is confirmed by the following proposition.
\begin{proposition}
  \label{prop:free-fphi}
  If $A$ is partial Boolean algebra, $\phi(p_1,\dots,p_n)$ a propositional formula, and $\alpha \in \mdomA{\phi(\vec{p})}$, then there exists a $\pBA$ morphism $\hat{\alpha}\colon F_{\phi} \rightarrow A$ such that $\hat{\alpha}([\phi]) = \phi^{A}(\alpha)$.
\end{proposition}
\begin{proof}
By structural induction on $\phi(\vec{p})$, we construct $\hat{\alpha}$ and prove that $\hat{\alpha}([\phi]) = \phi^{A}(\alpha)$.
For the base cases,
  \begin{enumerate}[label=(\arabic*)]
    \item Suppose $\phi = \top$, then $F_{\phi} = \Two$. By \ref{cond:ext}, every pBA contains $\Two$ as a Boolean subalgebra.
      Let $\hat{\alpha}\colon F_{\phi} \rightarrow A$ be the inclusion from the $\Two = F_{\phi}$ subalgebra into $A$.
      By Definition~\ref{def:substitution}, for all $\alpha \in \mdomA{\phi(\vec{p})} = A^{\vec{p}}$, $\phi^{A}(\alpha) = 1_A$.
      Since $\pBA$ morphisms preserve units, $\phi^{A}(\alpha) = 1_A = \hat{\alpha}(1_{F_{\phi}}) = \hat{\alpha}([\top])$.
    \item Suppose $\phi = p_i$, then $F_{\phi}$ is the four element Boolean algebra on the set $\{0,p_i,\neg p_i,1\}$.
      For all $\alpha \in \mdomA{\phi(\vec{p})}$, we obviously have a $\pBA$-morphism $\hat{\alpha}\colon F_{\phi} \rightarrow A$ generated by $\hat{\alpha}(p_i) = \alpha(p_i) = \phi^{A}(\alpha)$ since the image of the other elements in $F_{\phi}$ follow from preserving negation and units.
  \end{enumerate}
  For the inductive steps,
  \begin{enumerate}[resume,label=(\arabic*)]
    \item Suppose $\phi = \neg \psi$. By the inductive hypothesis, for every $\alpha \in \mdomA{\psi(\vec{p})} = \mdomA{\phi(\vec{p})}$, there is a $\pBA$ morphism $\hat{\alpha_{\psi}}\colon F_{\psi} \rightarrow A$ such that $\hat{\alpha}_{\psi}([\psi]) = \psi^{A}(\alpha)$.
      Since $F_{\phi} = F_{\psi}$, we can take $\hat{\alpha} = \hat{\alpha_{\psi}}$.
      By $\hat{\alpha}$ preserving negations, we obtain that $\hat{\alpha}([\phi]) = \hat{\alpha}([\neg \psi]) = \neg_{A} \hat{\alpha}([\psi]) = \neg_{A} \psi^{A}(\alpha) = \phi^{A}(\alpha)$.
    \item Suppose $\phi = \psi_1 \bowtie \psi_2$ for $\bowtie \in \{\wedge,\vee\}$.
      By hypothesis, $\alpha \in \mdomA{\phi(\vec{p})}$, so in particular $\alpha \in \mdomA{\psi_1(\vec{p})} \cap \mdomA{\psi_2(\vec{p})}$ and $\psi^A_1(\alpha) \odot_A \psi^A_2(\alpha)$.
      By the inductive hypothesis, for $i \in \{1,2\}$ there exist $\pBA$-morphisms $\hat{\alpha}_i \colon F_{\psi_i} \rightarrow A$ such that $\hat{\alpha}_i([\psi_i]) = \psi^A_i(\alpha)$.
      The construction of $F_{\phi}$ is in three steps: $H_{\phi} = F_{\psi_1} \uplus F_{\psi_2}$, $G_{\phi} = H_{\phi}[\circledcirc]$, and finally $F_{\phi} = G_{\phi}\backslash \circleddash$.
      These structures are equipped with morphisms $\eta_{\phi}\colon H_{\phi} \rightarrow G_{\phi}$, $\nu_{\phi}\colon G_{\phi} \rightarrow F_{\phi}$.
      From the universal property of the coproduct $H_{\phi}$, there exists a unique morphism $\hat{\alpha}_1 \uplus \hat{\alpha}_2 \colon H_{\phi} \rightarrow A$.
      By property~\ref{item:extension-existence} of $\circledcirc$-extension $G_{\phi}$ and $\psi^A_1(\alpha) \odot_A \psi^A_2(\alpha)$, there exists a morphism $h\colon G_{\phi} \rightarrow A$.
      By property~\ref{item:quotient-existence} of $\circleddash$-quotient $F_{\psi}$ and $\alpha \in \mdomA{\psi_1(\vec{p})} \cap \mdomA{\psi_2(\vec{p})}$, there exists a morphism $\bar{h}\colon G_{\phi} \rightarrow F_{\phi}$.
      Moreover, the collection of these morphisms are such that the following diagram commutes:
      \[
        \begin{tikzcd}
          H_{\phi} \ar[d,"{\eta_\phi}"'] \ar[r,"{\hat{\alpha}_1 \uplus \hat{\alpha}_2}"] & A \\
          G_{\phi} \ar[ur,"h"] \ar[d,"{\nu_\phi}"'] \\
          F_{\phi} \ar[uur,"\bar{h}"']
        \end{tikzcd}
      \]
    \end{enumerate}
    Thus, we can set $\hat{\alpha} = \bar{h}$.
    From chasing the diagram and $\phi = \psi_1 \bowtie \psi_2$, we can conclude that $\hat{\alpha}([\phi]) = \phi^{A}(\alpha)$.
  \end{proof}
  The pBA $F_{\phi}$ is the minimal pBA which has an element $[\phi] \in F_{\phi}$ corresponding to the formula $\phi$ such that if $\phi$ has variables amongst $\vec{p} = (p_1,\dots,p_n)$ there is a corresponding meaningful subsitution $\gamma$ where $\gamma(p_i) = [p_i] \in F_{\phi}$ and $\phi^{F_{\phi}}(\gamma) = [\phi]$.
  Thus, we can reduce $\All$-satisfiablity to weak satisfiablity in $F_{\phi}$ via $\gamma$.
  Similarly, we can consider $M_{\phi}$ where $M_{\phi}$ is the quotient of $F_{\phi}$ by the relation $\{([\phi],1_{F_{\phi}})\}$, and the corresponding meaningful subsitution $\chi$ where $\chi(p_i) = [[p_i]] \in M_{\phi}$ and $\phi^{M_{\phi}}(\chi) = 1 \in M_{\phi}$.
  In this case, we can reduce $\All$-satisfiablity to non-triviality and strong satisfiablity in $M_{\phi}$ via $\chi$.
\begin{proposition}
  \label{prop:all-sat-to-free}
  Let $\phi(\vec{p})$ be a propositional formula, then the following are equivalent:
  \begin{enumerate}[label=(\arabic*)]
    \item \label{item:all-sat-free-all} $\phi$ is $\All$-satisfiable
    \item \label{item:all-sat-free-weak} $\phi$ is weakly satisfied in $F_{\phi}$ via $\gamma$.
    \item \label{item:all-sat-free-strong} $\phi$ is strongly satisfied in $M_{\phi}$ via $\chi$ and $M_{\phi}$ is non-trivial.
  \end{enumerate}
\end{proposition}
\begin{proof}
  \ref{item:all-sat-free-all} $\Rightarrow$ \ref{item:all-sat-free-weak} Suppose for contradiction $\phi$ is $\All$-satisfiable and $\phi^{F_{\phi}}(\gamma) = 0_{F_{\phi}}$.

  By definition of $\All$-weak satisfiablity, there exists a non-trivial pBA $A$ and $\alpha \in \mdomA{\phi(\vec{p})}$ such that $\phi^{A}(\alpha) \not= 0_A$.
  By Proposition~\ref{prop:free-fphi}, there exists a $\pBA$-morphism such that $\hat{\alpha} \colon F_{\phi} \rightarrow A$ where $\hat{\alpha}([p_i]) = \alpha(p_i)$ and $\hat{\alpha}([\phi]) = \phi^A(\alpha) \not= 0_A$.
  On the other hand, since $[\phi] = \phi^{F_{\phi}}(\gamma) = 0_{F_{\phi}}$, then by $\hat{\alpha}$ preserving units, $\hat{\alpha}([\phi]) = \hat{\alpha}(0_{F_{\phi}})= 0_A$. Contradiction.

  \ref{item:all-sat-free-weak} $\Rightarrow$ \ref{item:all-sat-free-strong} If $[\phi]  = \phi^{F_{\phi}}(\gamma) \not= 0_{F_{\phi}}$ with $\gamma \in {F_{\phi}}^{\phi(\vec{p})}$ defined as $\gamma(p_i) = [p_i] \in F_{\phi}$, then $M_{\phi}$ is non-trivial.
  By construction of the quotient $M_{\phi}$, $\phi^{M_{\phi}}(\chi) = 1_{M_{\phi}}$ with $\chi \in {M_{\phi}}^{\phi(\vec{p})}$ defined as $\chi(p_i) = [[p_i]] \in M_{\phi}$.

  \ref{item:all-sat-free-all} $\Leftarrow$ \ref{item:all-sat-free-weak} If $M_{\phi}$ is non-trivial, then $M_{\phi} \in \All$ and $\phi$ is $\All$-satisfiable.
\end{proof}
In light of Proposition~\ref{prop:all-sat-to-free}, and $F_{\phi}$ and $M_{\phi}$ being finite, there is an algorithm for $\All$-satisfiability which constructs $M_{\phi}$ from $\phi(\vec{p})$, checks its non-triviality, then decides $\VARSAT$ on input $(M_{\phi},1_{M_{\phi}},\phi(\vec{p}))$.
This algorithm demonstrates that $\All$-satisfiability and its dual problem is decidable resolving an open question posed in~\cite{logicalQuantum}.
However, in the worst case, $F_{\phi}$ and $M_{\phi}$ are doubly-exponential in the size of the input $\phi(\vec{p})$.
Thus, such an algorithm demonstrates that $\SAT{\All} \in \DEXPTIME$.

We can improve upon this result by observing that it is not necessary to construct the full $\pBA$ $A$ to check that it is non-trivial and that it satisfies $\phi(\vec{p})$.
Instead, we can augment the witness of $\VARSAT$, i.e.\ the meaningful substitution $\alpha \in A^{\vec{p}}$ with some additional data in order to demonstrate that $A$ is non-trivial.
To this end, suppose $A$ is a pBA, $\phi(\vec{p})$ is a propositional formula, and $\alpha \in A^{\phi(\vec{p})}$ is a meaningful substitution, we consider the induced subgraph $G_{\alpha}$ of $(A,\odot_A)$ by the subformula witnesses arising from $\alpha$.
Explicitly, $G_{\alpha}$ is the induced subgraph of $(A,\odot_A)$ on the set:
\[ V(G_{\alpha}) := \{\alpha_{\psi} \mid \psi \in \Sub(\phi(\vec{p})) \}. \]
where we define $\alpha_{\psi} := \psi^{A}(\alpha)$ for a cleaner notation.
The graph $G_{\alpha}$ has a root $\alpha_{\phi}$ corresponding to the full formula $\phi(\vec{p})$.

Let $K$ be family of cliques of $G_{\alpha}$ which cover every edge of $G_{\alpha}$.
For a clique $C \in K$, let $\Eq(C)$ be all the height-$1$ Boolean equations satisfied in $\langle C \rangle_A$ involving the elements in $C$.
Note that the equation $\alpha_{\phi} = 1$ in $\Eq(C)$ if the root $\alpha_{\phi}$ of $G_{\alpha}$ is in $C$.
A function $\nu \colon C \rightarrow \{0,1\}$ respects a set $E$ of height-$1$ Boolean equations if for every equation $e \in E$, $\nu$ restricted to the variables of $e$ is the interpretation of $e$ in the two-element Boolean algebra $\Two$, e.g.\ if $e$ is $\alpha_{\psi} = \alpha_{\psi_1} \wedge \alpha_{\psi_2}$, then $\nu(\alpha_{\psi}) = \nu(\alpha_{\psi_1}) \wedge_{\Two} \nu(\alpha_{\psi_2})$.
These notions are connected together in the following lemma and allow us to provide a polynomial size witness to the non-triviality of $A$.
\begin{lemma}
  \label{lem:ntriv}
  Let $A$ be a $\pBA$, $\phi(\vec{p})$ a propositional formula, $\alpha \in A^{\phi(\vec{p})}$, and $K$ be a non-empty clique-edge cover of $G_{\alpha}$.
  The following are equivalent:
\begin{enumerate}[label=(\arabic*)]
  \item \label{item:whole-ntriv} $A$ is non-trivial and $\phi^{A}(\alpha) = 1_A$.
  \item \label{item:subalgebra-ntriv} For every clique $C \in K$, $\langle C \rangle_A$ is non-trivial and satisfies $\Eq(C)$.
  \item \label{item:filter-ntriv} For every clique $C \in K$, for every $a \not\leq_A b \in C$, there exists a function $\nu_{a,b}\colon C \rightarrow \{0,1\}$ which respects $\Eq(C)$, $\nu_{a,b}(a) = 1$ and $\nu_{a,b}(b) = 0$.
\end{enumerate}
\end{lemma}
\begin{proof}
  \ref{item:whole-ntriv} $\Rightarrow$ \ref{item:subalgebra-ntriv} If $A$ is non-trivial, then for every $C \in \mathcal{C}$, the smallest Boolean subalgebra $\langle C \rangle_A$ of $A$ containing $C$ must contain $\{0_A,1_A\}$.
  Since $A$ is non-trivial by hypothesis, $\{0_A,1_A\}$ is two-element set and $\langle C \rangle_A$ is non-trivial.
  By construction, the Boolean algebra $\langle C \rangle_A$ satisfies the equations in $\Eq(C)$.

  \ref{item:subalgebra-ntriv} $\Rightarrow$ \ref{item:filter-ntriv} Suppose $C \in \mathcal{C}$ and $a \not\leq_A b \in C$.
  By hypothesis, $\langle C \rangle$ is non-trivial.
  Thus, for $a \not\leq_A b \in C$, there exists a boolean algebra homomorphism $\mu_{a,b}\colon \langle C \rangle_A \rightarrow 2$ such that $\mu_{a,b}(a) = 1$ and $\mu_{a,b}(b) = 0$.
  We can then define $\nu_{a,b} = \mu_{a,b} \circ i$ where $i\colon C \rightarrow \langle C \rangle_A$ is the inclusion.
  Each of the function $\nu_{a,b}$ respect $\Eq(C)$ since $\mu_{a,b}\colon \langle C \rangle_A \rightarrow $ is a Boolean algebra homomorphism.

  \ref{item:subalgebra-ntriv} $\Leftarrow$ \ref{item:filter-ntriv} Suppose for every $C \in K$, there exists functions $\{\nu_{a,b}\}$ satisfying the conditions in \ref{item:filter-ntriv}.
  To show that $\langle C \rangle_A$ is non-trivial, it suffices to show that there exists an atom $z \in \langle C \rangle_A$ since atoms are by definition non-zero elements.
  To construct $z \in \langle C \rangle_A$, we note that every atom in $\langle C \rangle_{A}$ is equal to the meet of elements $d \in \langle C \rangle_{A}$ which are either generators $c \in C$ or complements $\overline{c}$ of generators.
  Let $S \subset C$ be the set of generators appearing (with or without negation) in the expression of an atom $z$.
  Without loss of generality, we can assume the generators which appear in $S$ form an anti-chain with respect to $\leq_A$ in $C$.
  Namely, if $c \leq_A d$ and both $c,d$ are generators in $S$, then because $c \wedge d = c$, any instance of $d$ in $S$ can be eliminated and any instance of $\overline{d} \in S$ would imply
  \[ z \leq c \wedge \overline{d} = (c \wedge d) \wedge \overline{d} = c \wedge (d \wedge \overline{d}) = 0 \]
  which would constradict the supposition that $z$ is an atom.
  We arbitrarily enumerate the anti-chain $s_1 \not\leq_{A} \dots \not\leq_{A} s_n$.
  Since every pair $s,t \in S$ is such that $s \not\leq_A t$, by hypothesis there is a function $\nu_{s,t}\colon C \rightarrow \{0,1\}$ which sends either $s$ or $t$ to $1$ and the other to $0$.
  Let $I \subset n$ be the subset of indices such that $\nu_{s_{i},s_{i+1}}(s_i) = 1$.
  Thus, we can then define $z = \bigwedge_{i \in I} s_i \wedge \bigwedge_{i \not\in I} \overline{s_i}$.
  Therefore $\langle C \rangle_A$ is non-trivial.
  From each function $\nu_{a,b}\colon C \rightarrow \{0,1\}$, we obtain a $\BA$-morphism $\hat{\nu_{a,b}}\colon F(C) \rightarrow \Two$ where $F(C)$ is the free Boolean algebra on elements.
  Since each function $\nu_{a,b}\colon C \rightarrow \{0,1\}$ respects $\Eq(C)$, through quotienting we obtain a $\BA$-morphism $\mu_{a,b}\colon \langle C \rangle_A \rightarrow \Two$.
  It follows that $\langle C \rangle_A$ satisfies $\Eq(C)$.

  \ref{item:whole-ntriv} $\Leftarrow$ \ref{item:subalgebra-ntriv} Since $K$ is non-empty, there exists some clique $C \in K$.
  By hypothesis, the subalgebra $\langle C \rangle_A$ is non-trivial, so $A$ must be non-trivial.
  It follows from every $\langle C \rangle_A$ satisfying  $\Eq(C)$ that $\phi^{A}(\alpha) = 1_A$
  Namely, if we recursively unpack the definition of $\phi^{A}\colon A^{\phi(\vec{p})} \rightarrow A$, we must verify at each subformula $\psi$ the corresponding height-$1$ equation $e_{\psi}$ is satisfied in $A$.
  Since $e_{\psi} \in \Eq(C)$ for some clique $C \in K$, then by hypothesis $e_{\psi}$ is satisfied in $\langle C \rangle_A$, and thus $A$.
\end{proof}

\begin{theorem}\label{thm:allsat}
$\SAT{\All}$ is $\NP$-complete.
\end{theorem}
\begin{proof}
  Hardness follows from Corollary~\ref{cor:csat-np-hard}.
  To show membership in $\NP$, we need to produce a certificate of non-triviality of some pBA $A$ and strong-satisfiablity of the input $\phi$ via a meaningful substitution $\alpha \in A^{\phi(\vec{p})}$.
  The problem is the pair $(A,\alpha)$ is unlikely to be polynomial size.
  Instead, we use a certificate consisting of the following data which gives a partial view of the objects $A$ and $\alpha$:
  \begin{itemize}
    \item A compatiblity graph $G$ which view as isomorphic to $G_{\alpha}$ of size $\leq |\phi(\vec{p})|$
    \item A binary relation $\leq$ on $V(G)$ which view as a restriction of $\leq_A$ to $V(G_{\alpha})$ and pulled back along the isomorphism $G \cong G_{\alpha}$.
    \item A non-empty clique edge cover $K$ of $G$ which encode as at most $|E(G)|$ bitstrings of length $|V(G)|$.
    \item For every $C \in K$, a set of equations $\Eq(C)$ of size $\leq |\Sub(\phi(\vec{p}))|$.
    \item For every $C \in K$, a set of at most $|C|^2$ functions $\nu_{a,b}\colon C \rightarrow \{0,1\}$ (one for every pair $a \not\leq b$) encoded as bitstrings of length $|C|$.
  \end{itemize}
  The size bounds demonstrate that our certificate is polynomial in the length of $\phi(\vec{p})$.
  By Lemma~\ref{lem:ntriv}, these data are sufficient to prove non-triviality of $A$ and satisfiability of $\phi$ via $\alpha$.
  Checking the conditions in Lemma~\ref{lem:ntriv}\ref{item:filter-ntriv} can be done in polynomial time.
  $\SAT{\All}$ is $\NP$-hard by $\One \not\in \All$ and Proposition~\ref{prop:sat-reduction}.
\end{proof}
\section{Kochen-Specker Proofs}
\label{sec:kochen-specker}
Partial Boolean algebras were defined by Kochen and Specker in order to formalise a necessary condition for non-contextual hidden-variable theories of quantum theory to exist.
This condition involved the existence of a $\pBA$ morphism from pBAs arising in quantum theory to the two-element Boolean algebra $\Two$.
Thus, to begin we define these motivating pBAs arising in quantum theory.

Recall that in the standard textbook formulation of quantum theory the pure states of a quantum system live in a vector space $\Hilb$ over the field of complex numbers $\mathbb{C}$.
Measurements of a quantum system with pure states in $\Hilb$ are identified with the bounded self-adjoint, also called Hermitian, linear maps $M\colon \Hilb \rightarrow \Hilb$.
The operation of measuring $M$ on a state yields one of the eigenvalues $\lambda$ of $M$ and as $M$ is self-adjoint, $\lambda \in \mathbb{R}$.
After measurement, the state of the system collapses to one of the eigenvectors corresponding to the eigenvalue $\lambda$.
We denote the set of bounded self-adjoint linear maps over $\Hilb$ as $\Bnd(\Hilb)$.
By the spectral theorem, each outcome $\lambda$ of the measurement $M \in \Bnd(\Hilb)$ can be identified with the idempotent element $E_{\lambda} \in \Bnd(\Hilb)$ which projects onto the eigenspace of the eigenvalue $\lambda$ of $M$.
We denote the set of idempotent self-adjoint linear maps over $\Hilb$ as $\Pj(\Hilb)$.
Projectors in $\Pj(\Hilb)$ have eigenvalues $0,1$, i.e.\ answer yes/no questions, and are in bijection with subspaces of $\Hilb$.
Thus, `measuring' the yes/no question $E_{\lambda}$ yields an answer to the question: Did measurement $M$ have outcome $\lambda$?
Given a projector $E \in \Pj(\Hilb)$, let $\im(E)$ denote its image subspace and $\rank(E)$ denote the dimension of $\im(E)$, i.e.\ the \emph{rank} of $E$.
\begin{definition}
\label{def:projector-pba}
Given a real or complex Hilbert space $\Hilb$, the \emph{$\Hilb$-projector partial Boolean algebra} is the set $Q = \Pj(\Hilb)$ where:
\begin{itemize}
  \item $1_Q$ is the projector $I_{\Hilb}$ onto the whole space $\Hilb$, $0_Q$ is the projector $0_{\Hilb}$ onto the $0$-dimensional subspace of $\Hilb$,
  \item $E \odot_Q F$ if $[E,F] = EF - FE = 0_{H}$, i.e.\ two projectors are commeasurable if they commute,
  \item $\neg_{Q} E = \ocomp{E} = I_{\Hilb} - E$,
  \item $E \vee_{Q} F = E + F - EF$,
  \item $E \wedge_{Q} F = EF$.
\end{itemize}
\end{definition}
Pure states of a quantum system are vectors (up to normalisation) and thus can be recovered from $\Pj(\Hilb)$ as the set of rank-$1$ projectors, denoted $\Pj_1(\Hilb)$.

For a non-contextual hidden variable theory of the quantum system described $\Hilb$ to exist, there would necessarily be a $\pBA$-morphism $f\colon \Pj(\Hilb) \rightarrow \Two$.
Thus, the Kochen-Specker theorem demonstrates that there cannot exist such a $\pBA$-morphism for real or complex Hilbert space $\Hilb$ of dimension $\geq 3$.
Proofs of Kochen-Specker's theorem can be formulated as exhibiting an orthogonality graph $G$ of vectors in $\Hilb$ which lacks a certain two colouring of the vertices.

We make the notion of Kochen-Specker proof precise using a formalisation which differs slightly from the usual graph-theoretic formalisations, see e.g.\ \cite{gadgetKS} for details, by allowing vertices to be assigned to non-unital projectors and not necessarily injectively.
In the following definitions, a graph $G$ is simple, undirected, without self-loops.
A graph $G$ has vertex set $V(G)$, edge set $E(G)$, and set of maximum cliques $\Omega(G)$.
\begin{definition}
  Given a graph $G$ and real or complex Hilbert space $\Hilb$, an \emph{orthogonal assignment $f\colon G \rightarrow \Pj(\Hilb)$} is a function $f\colon V(G) \rightarrow \Pj(\Hilb)$ if the following two conditions hold:
  \begin{enumerate}[label=(O\arabic*)]
    \item \label{item:oa-edge-condition} For every $(v,w) \in E(G)$, $f(v)f(w) = 0_{\Hilb}$
    \item \label{item:oa-clique-condition} For every maximum clique $C \in \Omega(G)$ of $G$, $\sum_{v \in C} f(v) = I_{\Hilb}$.
  \end{enumerate}
  We say $f$ is a \emph{rank-$1$ orthogonal assignment} if $f \colon V(G) \rightarrow \Pj_1(\Hilb)$ and use the notation $f\colon G \rightarrow \Pj_1(\Hilb)$.
\end{definition}
An orthogonal assignment is a generalisation of the notion of non-contextual colouring.
\begin{definition}
  Given a graph $G$, a function $f\colon V(G) \rightarrow \{0,1\}$  is a \emph{non-contextual colouring of $G$} if the following two conditions hold:
  \begin{enumerate}[label=(C\arabic*)]
    \item \label{item:nc-edge-condition} For every edge $(v,w) \in E$, $f(v) + f(w) \leq 1$
    \item \label{item:nc-clique-condition} For every maximum clique $C \in \Omega(G)$ of $G$, $\sum_{v \in C} f(v) = 1$.
  \end{enumerate}
\end{definition}
\begin{definition}
  \label{def:kochen-specker-proof}
  For a (possibly infinite) cardinal $d$, a graph $G$ is a \emph{$d$-dimensional Kochen-Specker proof} if $G$ has an orthogonal assignment $f\colon G\rightarrow \Pj(\Hilb)$ to a real or complex $d$-dimensional Hilbert space $\Hilb$ and $G$ does \textbf{not} have a non-contextual colouring.
\end{definition}

If $S \subseteq \Hilb$ is a set vectors, we can form its orthogonality graph $G_S$ with vertices $S$ and $\vec{v},\vec{v'} \in S$ are adjacent if $\vec{v}^{\dagger}\vec{v'} = 0$, i.e.\ are orthogonal.
Obviously, $G_S$ has a rank-$1$ orthogonal assignment $f_{S}(\vec{v}) = E_{\vec{v}}$.
We say $S$ is basis-complete if it is equal to the union of a family of orthonormal bases $\mathcal{B}_{S}$.
For any basis-complete set $S \subseteq \FK^d$, the orthogonality graph $\omega(G_S) = d$ and every maximal clique is a maximum clique.
Appropriating terminology from topology, we will say a graph $G$ is a \emph{$d$-facet graph} if $\omega(G) = d$ and every maximal clique is a maximum clique.
If $S$ is basis-complete, then $G_S$ is a facet graph, so condition~\ref{item:nc-edge-condition} is redundant and we obtain the following equivalence:
\begin{proposition}
\label{prop:bc-colouring}
If $S$ is basis-complete, then the following are equivalent:
\begin{enumerate}[label=(\arabic*)]
  \item \label{item:bc-nc-colouring} $G_S$ does not have a non-contextual colouring.
  \item \label{item:bc-2c-colouring} $G_S$ has no colouring $f\colon S \rightarrow \{0,1\}$ satisfying condition \ref{item:nc-clique-condition}.
  \item \label{item:bc-mc-clique} For every function $f\colon S \rightarrow \{0,1\}$, there exists a maximum (equiv. maximal) clique $C_{f}$ of $G_{S}$ such that $\sum_{c \in C_f} f(u) \not= 1$.
\end{enumerate}
\end{proposition}
\begin{proof}
  $\ref{item:bc-nc-colouring} \Rightarrow \ref{item:bc-2c-colouring}$ By contrapositive.
  Suppose $f$ satisfies \ref{item:nc-clique-condition}.
  It remains to show that $S$ being basis complete implies $f$ satisfies \ref{item:nc-edge-condition}.
  By $S$ being basis-complete, for any orthogonal pair of vectors $v \sim w$, there is a basis $B$ where $\{v,w\} \subseteq B$.
  Basis $B$ induces a maximum clique in $G_S$. Therefore,
  \[ f(v) + f(w) \leq \sum_{u \in B} f(u) = 1. \]
  $\ref{item:bc-nc-colouring} \Leftarrow \ref{item:bc-2c-colouring}$ By contrapositive. Any non-contextual colouring of $G$ satisfies \ref{item:nc-clique-condition}.
  $\ref{item:bc-2c-colouring} \Leftrightarrow \ref{item:bc-mc-clique}$ Item~\ref{item:bc-mc-clique} is simply unpacking the negation of \ref{item:nc-clique-condition}.
\end{proof}
The original proof of the Kochen-Specker theorem~\cite{hiddenVariableQM} was a proof in the sense of Definition~\ref{def:kochen-specker-proof} with dimension 3.
This proof was constructed as the orthogonality graph $G_S$ of a set of 117 vectors $S \subseteq \mathbb{R}^3$ (or 120 vectors, if we take the basis completion of $S$).
Subsequent work has resulted in Kochen-Specker proofs of lower cardinality.
The next section involves a reduction which uses a basis-complete $3$-dimensional Kochen-Specker proof $G$ as a gadget.
This reduction is entirely independent in the choice of basis-complete Kochen-Specker proof.
Therefore, we opt to consider the smallest known example in dimension $3$, the basis-completion of the Conway-Kochen set from~\cite{conwayKochenSet} consisting of only 40 vectors which we denote as $\CKS$.

The contrapositive of \cite[Theorem 4]{hiddenVariableQM} states that the non-existence of $\pBA$ morphism $h\colon A \rightarrow \Two$ is equivalent to exhibiting a classical propositional contradiction $\varphi$ which is satisfied in the pBA $A$.
It follows that from a finite $d$-dimensional Kochen-Specker proof $G$, we should be able to construct a propositional contradiction $\varphi_G$ which is satisifed in $\Pj(\Hilb)$.
For any finite graph $G$, we can associate a propositional formula $\varphi_{G}$:
\[ \varphi_{G}(\vec{p}) := \bigwedge_{(v,v') \in E(G)} \neg (p_{v} \wedge p_{v'}) \wedge \bigwedge_{C \in \Omega(G)} \bigvee_{v \in V(C)} p_{v}.  \]
where variables $[\vec{p}]$ are indexed by $V(G)$.
The desired connection of $\varphi_G$ to a Kochen-Specker proof $G$ is verified in the following proposition.
\begin{proposition}
\label{prop:ortho-assignments}
Let $Q = \Pj(\Hilb)$ for some real or complex Hilbert space.
$G$ has an orthogonal assignment $f\colon G \rightarrow Q$ iff $\varphi_G(\vec{p})$ is strongly satisifed in $Q$.
In paticular, $G$ has non-contextual colouring iff $\varphi_G(\vec{p})$ is classically satisfiable.
\begin{proof}
  The conjunction $\bigwedge_{(v,v')} \neg (p_v \wedge p_{v'})$ asserts that $f(v)f(v') = 0$, i.e.\ condition~\ref{item:oa-edge-condition}.
  Using $f(v)f(v') = 0$, the clause $\bigvee_{v \in V(C)} p_{v}$ requires that $\sum_{i \in \setd} f(v) = 1$, i.e\ condition~\ref{item:oa-clique-condition}.
  In this case of $Q = \Pj(\mathbb{K}) \cong \Two$, conditions~\ref{item:oa-edge-condition}-\ref{item:oa-clique-condition} are simply conditions~\ref{item:nc-edge-condition}-\ref{item:nc-clique-condition} of a non-contextual colouring.
\end{proof}
\end{proposition}

\section{Fixed dimension quantum satisfiability}
\label{sec:fixed-quantum-sat}
In this and the next section, we establish complexity bounds for the satisfiability problem for partial Boolean algebras which arise from the study of quantum contextuality.
This section focuses on the case for quantum system whose states reside in a $d$-dimensional Hilbert space for some $d \in \mathbb{N}$.

Since we restrict ourselves to the finite-dimensional case in this section, we take $\Hilb = \FK^d$.
In this case $\Bnd(\FK^d)$ is the set of $d \times d$ matrices $M \in M_{\FK}(d,d)$ which satisfy the self-adjoint condition: $M$ is equal to its conjugate transpose $M^{\dagger}$.
The partial Boolean algebra $Q = \Pj(\FK^d)$ consists of $d \times d$ projection matrices, $I_{\Hilb}$ is the identity matrix, $0_{\Hilb}$ is the all-zeros matrix, and the operations used in Definition~\ref{def:projector-pba} for defining $\neg_{Q}$, $\vee_{Q}$ and $\wedge_{Q}$ are matrix addition and multiplication.

Recall that each measurement has an associated set of orthogonal projectors called a PVM.
In the maximal case, this PVM is associated to a basis of the Hilbert space $\Hilb$.
We define the formula $\basis_d(p_1,\dots,p_d) := \varphi_{K_d}$ where $K_d$ is the $d$-vertex complete graph to `capture' PVMs in $\Hilb$.
\begin{proposition}
\label{prop:basis-capture}
Let $Q = \Pj(\FK^d)$, if $E_1,\dots,E_d \in Q$ are non-zero projectors, then $\basis_{d}^Q(E_1,\dots,E_d) = 1_Q$ if, and only if,
\begin{enumerate}[label=(\arabic*)]
  \item for every $i \in \setd$, there exists unit vectors $\vec{v_i} \in \FK^d$ where $E_i = \vec{v_i} \vec{v_i}^{\dagger}$, and
  \item $\{\vec{v_1},\dots,\vec{v_d}\}$ is an orthonormal basis of $\FK^d$.
\end{enumerate}
In particular, $(p_1,\dots,p_d) \mapsto (E_1,\dots,E_d)$ is a rank-$1$ assignment.
\end{proposition}
\begin{proof}
  The conjunction $\bigwedge_{i \not= j} \neg (p_i \wedge p_j)$ asserts that $E_iE_j = 0$.
  Thus, $\{E_1,\dots,E_d\}$ is an orthonormal set.
  Using $E_iE_j = 0$, the clause $p_1 \vee \dots \vee p_d$ asserts that $\sum_{i \in \setd} E_i = 1$.
  As each $E_i$ is distinct and non-zero, since the $\sum E_i = 1$ is $\Pj(\FK^d)$, we can conclude that every $E_i$ is a rank-$1$ projector.
  Thus, there exists vectors $\vec{v_i}$ such that $E_i = \vec{v_i}\vec{v_i}^{\dagger}$.
  From the orthonormality of $\{E_1,\dots,E_d\}$, we can conclude that $\{\vec{v_1},\dots,\vec{v_d}\}$ is also an orthonormal set.
  From the sum $\sum_{i \in \setd} E_i = 1$, we can conclude that $\{\vec{v_1},\dots,\vec{v_d}\}$ spans $\FK^d$.
\end{proof}

Our first observation is that for $d = 1$ or $d = 2$, the weak and strong satisfiability problems are $\NP$-complete.
\begin{theorem}
  If $d \in \{1,2\}$ and $\FK \in \{\mathbb{R},\mathbb{C}\}$, then problems $\strongSAT{\Pj(\FK^d)}$ and $\weakSAT{\Pj(\FK^d)}$ are $\NP$-complete.
\end{theorem}
\begin{proof}
  In all cases, the problems are equivalent to the classical satisfiability problem $\SAT{\Total}$.
  For the case where $d = 1$, we note that $\Pj(\mathbb{R}) = \Pj(\mathbb{C}) = \Two$ is the two-element Boolean algebra.

  For the case where $d = 2$. By Theorem 4 of~\cite{hiddenVariableQM}, a $\pBA$ $M$ has a $\pBA$-morphism $f\colon M \rightarrow \Two$ if, and only if, every formula which is strongly satisfiable in $M$ is classically satisfiable.
  In particular, since $\Pj(\FK^2)$ has (uncountably many) morphisms to $\Two$, i.e.\ by choosing which element in a pair of orthogonal rank-1 projectors is assigned to $1$, every input to $\strongSAT{\Pj(\FK^2)}$ is classically satisfiable.
For $\weakSAT{\Pj(\FK^2)}$, every yes-instance $\phi(\vec{p})$ of $\weakSAT{\Pj(\FK^2)}$ is either a yes-instance of $\strongSAT{\Pj(\FK^d)}$ or the witness $\alpha$ is such that $\phi^{\Pj(\FK^2)}(\alpha)$ is a rank-$1$ projector -- and therefore a yes-instance to $\strongSAT{\Pj(\FK)} \cong \Two$.
  Thus, all cases reduce to and from the classical satisfiablity problem $\SAT{\Total}$.
\end{proof}

\subsection{Real case}
Here, we prove that for every $d \geq 3$, the complexity of satisfiability in $\Pj(\mathbb{R}^d)$ is $\ER$-complete.

Thus, we began by defining the complexity class called the existential theory of the reals which is denoted $\ER$.
Instead of defining $\ER$ directly, we give a more general definition that is helpful in subsequent sections.

We associate a complexity class $\exists (T,M)$ to any first-order signature $L$ which contains the signature of rings $\Lrings = \{0,1,+,*\}$ and $L$-structure $M$.
For every such pair $(L,M)$, we define
\begin{itemize}
  \item $E(L)$ to be the set of first-order sentences $\Phi$ in $L$ of the form
    \[ \exists X_1, \dots, \exists X_n \Psi(X_1,\dots,X_n) \]
    where $\Psi(X_1,\dots,X_n)$ is a quantifier-free formula; and
  \item $E(L,M)$ to the set of sentences $\Phi \in E(L)$ which are true in $M$.
\end{itemize}
We can then define the complexity class $\exists (L,M)$ as the class of decision problems which are many-one polynomial-time reducible to  $E(T,M)$.
For example, $\NP$ can be defined as $\exists(L,M)$ where $L = \Lrings$ and $M$ is the two-element field $(\mathbb{Z}_2,0,1,\xor,\wedge)$.
The class $\ER$ is $\exists (L,M)$ where $L = \Lrings \cup \{\leq\}$ is the signature of \emph{ordered} rings and $M$ is the ordered field of real numbers ${(\mathbb{R},0,1,+,*,\leq)}$.
As usual, a decision problem $D$ is \emph{$\ER$-hard} if every problem in $\ER$ polynomial-time many-one reduces to $D$ and $D$ is \emph{$\ER$-complete} if $D$ is in $\ER$ and $\ER$-hard.
Many natural computational problems, usually related to discrete objects embedded in continuous spaces, are known to be $\ER$-complete, e.g.\ see the compendium~\cite{erCompendiumArxiv} for details.
The class $\ER$ can also be defined as the Boolean fragment of non-deterministic polynomial time on real Blum-Shub-Smale machines~\cite{bssRealMachines}.
In relation to classic complexity classes, it is known that $\NP \subseteq \ER \subseteq \PSPACE$.
The inclusion $\NP \subseteq \ER$ follows from the Cook-Levin theorem and the polynomial-time mapping, e.g.\ $\neg p \vee q \mapsto \exists P \exists Q(P = 0 \vee Q = 1)$, from a propositional formula to an existential sentence in the language of fields.
The inclusion of $\ER \subseteq \PSPACE$ was a seminal result proved by Canny~\cite{erPSPACE}.

Proposition~\ref{prop:realdsat-er} below demonstrates that the weak and strong satisfiability problems for $\Pj(\mathbb{R}^d)$ are in $\ER$.
This is because, using Definition~\ref{def:projector-pba}, a propositional formula $\phi(p_1,\dots,p_n)$ being satisfied in $\Pj(\mathbb{R}^d)$ translates into the existence of self-adjoint solutions to a set $S$ of (in)equations in the ring of matrices $M_{\mathbb{R}}(d,d)$.
Hence, by unfolding the definitions of matrix addition and multiplication in $M_{\mathbb{R}}(d,d)$ in terms of their $d^2$ entries, these equations translate to $O(d^2|S|)$ equations over the field $\mathbb{R}$.
\begin{proposition}
  \label{prop:realdsat-er}
  For every $d \geq 1$, $\strongSAT{\Pj(\mathbb{R}^d)}$ and $\weakSAT{\Pj(\mathbb{R}^d)}$ are in $\ER$.
\end{proposition}
\begin{proof}
  For every propositional formula $\phi(p_1,\dots,p_n)$, we can construct in polynomial time a existential first-order sentence $\Psi^{s}$ in the language $\Lrings \cup \{\dagger\}$ of rings with a unary dagger function symbol $\dagger$, and an existential first-order sentence $\Phi^{s}$ in the language of rings $\Lrings$, such that the following are equivalent:
\begin{enumerate}[label=(\arabic*)]
  \item \label{item:realdsat-er-pba} $\phi(p_1,\dots,p_n)$ is strongly satisfied in $\Pj(\mathbb{R}^d)$
  \item \label{item:realdsat-er-matrix}$\Psi^{s}$ is satisfied in $M_{\mathbb{R}}(d,d)$
  \item \label{item:realdsat-er-field} $\Phi^{s}$ is true in the ordered field $\mathbb{R}$.
\end{enumerate}
For the construction of $\Psi$, the variables are:
\begin{align*}
  [\vec{P}] &= \{P_1,\dots,P_n\}, \\
  [\vec{Z}] &= \{Z_{\psi} \mid \psi \text{ is a non-variable subformula of $\phi$}\}.
\end{align*}
We define $\Psi^{s}$ as the existential closure of the conjunction of these equations:
\begin{align*}
  P_l^2 &= P_l & \text{for $l \in \setn$} \\
  P_l &= P_l^{\dagger} & \text{for $l \in \setn$} \\
  Z_{\psi} &= 1 + (0 - V_{\psi'}) & \text{if } \psi = \neg \psi' \\
  Z_{\psi} &= V_{\psi_1} V_{\psi_2} \text{ and } [V_{\psi_1},V_{\psi_2}] = 0 & \text{if } \psi = \psi_1 \wedge \psi_2 \\
  Z_{\psi} &= V_{\psi_1} + V_{\psi_2} - V_{\psi_1}V_{\psi_2} \text{ and } [V_{\psi_1},V_{\psi_2}] = 0 & \text{if } \psi = \psi_1 \vee \psi_2 \\
  Z_{\phi} &= 1
\end{align*}
where $V_{p_l} = P_l$, $V_{\psi} = Z_{\psi}$ for every non-variable subformula $\psi$ of $\phi$, and $[E,F] = EF - FE$ is the commutator operation.
The first equation expresses each element $P_l$ is a self-adjoint projector.
The middle equations express Definition~\ref{def:projector-pba} of the partial operations $\neg_Q,\vee_Q,\wedge_Q$ in $Q = \Pj(\mathbb{R}^d)$.
In particular, the commutator equations ensure that the partial operations are well-defined.
Thus, if $E_1,\dots,E_n \in Q$ are such that $\phi^Q(E_1,\dots,E_n) = 1_A$, we can conclude that $\hat{\Psi}$ is satisfied in $M_{\mathbb{R}}(d,d)$ via the assignment $P_l \mapsto E_l$ and $Z_{\psi} \mapsto \psi^{Q}(E_1,\dots,E_n)$.
Conversely, if there exists an assignment of $\Gamma \colon [\vec{Z}] \cup [\vec{P}] \rightarrow M_{\mathbb{R}}(d,d)$ of variables to matrices in $M_{\mathbb{R}}(d,d)$ which satisfy $\Psi^{s}$, we can construct a meaningful substitution $\gamma \in \mdom{Q}{\phi}$ where $\gamma(p_l) = \Gamma(P_l)$ and $\psi^{Q}(\gamma) = \Gamma(Z_{\psi})$.
The last equation $Z_{\phi} = 1$ implies that $\phi^Q(\gamma) = 1_A$.
Thus, we have shown the equivalence $\ref{item:realdsat-er-pba} \Leftrightarrow \ref{item:realdsat-er-matrix}$.

We construct $\Phi^{s}$ from $\Psi^{s}$ by unfolding the definition of matrix addition multiplication, and transpose as they appear in the equations of $\Psi^{s}$.
For every variable $Z_{\psi}$ and $P_{l}$ of $\Psi^{s}$, $\Phi^{s}$ has $d^2$ many existentially-quantified variables $(Z_{\psi}^{ij})$ and $(P^{ij}_l)$.
Every equation (without the dagger operation) in $\Psi^{s}$ corresponds to a conjunction of $d^2$ equations in the quantifier-free part of $\Phi^{s}$ expressing the $ij$-th entry of the resulting matrix variable.
Additionally, for every variable $P_{l}$ of $\Psi^{s}$, we have an $d^2$-many equations $P^{ij}_{l} = P^{ji}_{l}$ in the quantifier-free part of $\Psi^{s}$ to express that the matrix assigned to $P_l$ is self-adjoint.
Thus, by construction of $\Phi^{s}$, the equivalence $\ref{item:realdsat-er-matrix} \Leftrightarrow \ref{item:realdsat-er-field}$ holds.

The size of $\Phi^{s}$ is linear in the size of $\phi(p_1,\dots,p_n)$.
Thus, the construction of $\Phi^{s}$ and the equivalence $\ref{item:realdsat-er-pba} \Leftrightarrow \ref{item:realdsat-er-field}$ demonstrates that $\strongSAT{\Pj(\mathbb{R}^d)}$ is in $\ER$.
We can similarly construct a sentence $\Phi^{w}$, where the equation $Z_{\phi} = 1$ in $\Psi^{s}$ is replaced with the inequality $Z_{\phi} \not= 0$, to demonstrate that $\weakSAT{\Pj(\mathbb{R}^d)}$ is in $\ER$.
\end{proof}

Now we proceed with the proof that $\strongSAT{\Pj(\mathbb{R}^d)}$ and $\weakSAT{\Pj(\mathbb{R}^d)}$ are $\ER$-hard for any fixed $d \geq 3$.
To accomplish this, we work by induction on dimension $d$.
Hence, for the base case, we first prove that $\SAT{\Pj(\mathbb{R}^3)}$ is $\ER$-hard.
We construct reductions from the $\ER$-complete problem from~\cite{crossProductER} called cross-product term satisfiability over the real projective plane $\RPTwo$.
We denote this decision problem $\XSAT(\RPTwo)$.
In order to describe~$\XSAT(\RPTwo)$, recall that the right-handed cross product is a binary operation $\cross : \RAThree \times \RAThree \rightarrow \RAThree$ defined coordinate-wise as:
\begin{equation}
  \label{eq:cross-product}
  \cross (\vec{v},\vec{w}) = (v_2w_3 - v_3w_2,v_3w_1 - v_1w_3,v_1w_2 - v_2w_1) .
\end{equation}
for every $\vec{v} = (v_1,v_2,v_3),\vec{w} = (w_1,w_2,w_3) \in \RAThree$.
The operation is defined so that if $w = u \cross v$, then $w$ is orthogonal to both $u$ and $v$.
This is a key fact which we exploit in our reduction.
The operation $\cross$ extends to a partial operation, which we also denote as $\cross$, on the real projective plane $\RPTwo$.
Elements of $\RPTwo$ are the lines $\mathbb{R}\vec{v} = \{\lambda \vec{v} \mid \lambda \not= 0\}$ for every $\vec{v} \in \RAThree \backslash \{\vec{0}\}$ in $\RAThree$, and $\cross$ is defined on representatives as $\mathbb{R}\vec{v} \cross \mathbb{R}\vec{w} = \mathbb{R}(\vec{v} \cross \vec{w})$ for distinct lines $\mathbb{R}v \not= \mathbb{R}w$.
A \emph{cross product term $t(x_1,\dots,x_n)$} with variables amongst $X = \{x_1,\dots,x_n\}$ is recursively defined as either one of the variables $x_i$ or $s \cross s'$ for cross product terms $s,s'$ with variables also amongst $X$.
Given a cross product term $t(x_1,\dots,x_n)$ and $\mathbb{R}\vec{v}_1,\dots,\mathbb{R}\vec{v}_n \in \RPTwo$, the value $\langle t(\mathbb{R}\vec{v}_1,\dots,\mathbb{R} \vec{v}_n) \rangle \in \RPTwo$ is defined by inductively applying the definition of $\cross$.
A cross product term $t(x_1,\dots,x_n)$ is \emph{satisfied in $\RPTwo$} if there exists $\mathbb{R}\vec{v}_1,\dots,\mathbb{R}\vec{v}_n \in \RPTwo$ such that $\langle t(\mathbb{R}\vec{v}_1,\dots,\mathbb{R}\vec{v}_n) \rangle$ is defined and $\langle t(\mathbb{R}\vec{v}_1,\dots,\mathbb{R}\vec{v}_n) \rangle = \mathbb{R}\vec{v}_1 \not= 0$.
The $\XSAT(\RPTwo)$ problem takes as input a cross product term $t(x_1,\dots,x_n)$ and decides if $t$ is satisfied in $\RPTwo$.
In order to reduce $\XSAT(\RPTwo)$ to $\SAT{\Pj(\RAThree)}$, we give a translation from a cross product term $t(x_1,\dots,x_n)$, to a propositional formula $\theta^{\mathbb{R}}_{t}$.

Consider a non-variable subterm $s = r \cross o$ of $t$.
The lines $\mathbb{R}\langle s \rangle$, $\mathbb{R}\langle r \rangle$, and $\mathbb{R}\langle o \rangle$ are in bijection with the rank-$1$ projectors $E_{\langle s \rangle}, E_{\langle r \rangle}$, and $E_{\langle o \rangle}$ onto these lines.
Since the line $\mathbb{R}\langle s \rangle$ is orthogonal to both $\mathbb{R}\langle r \rangle$ and $\mathbb{R}\langle o \rangle$, we need to express that the pairs $(E_{\langle s \rangle}, E_{\langle r \rangle})$ and $(E_{\langle s \rangle}, E_{\langle o \rangle})$ extend to triples of projectors, e.g.\ $(E_{\langle s \rangle}, E_{\langle r \rangle}) \mapsto (E_{\langle s \rangle}, E_{\langle r \rangle}, \ocomp{E_{\langle r \rangle} + E_{\langle s \rangle}})$ which correspond to orthonormal bases of $\RAThree$.

From these observations, the naive construction would be to define for every non-variable subterm $s = r \cross o$, a formula $\vartheta_{s}$ which is a conjunction of two $\basis_3$ formulas.
A witness that the term $t(x_1,\dots,x_n)$ is satisfied in $\RPTwo$ would then yield a meaningful substitution demonstrating that the conjunction of $\vartheta_s$ is in $\strongSAT{\Pj(\RAThree)}$.
However, such a meaningful substitution is not sufficient to produce a witness to $t$.
We need to show that any meaningful substitution of the conjunction produces a substitution where every variable is assigned to a rank-$1$ projector.
We can show that for each $\basis_3$ conjunct of $\vartheta_s$, there exists uncountably many rank-$1$ assignments (one for each orthonormal basis of $\RAThree$).
The problem arises when taking a conjunction of $\basis_3$ formulas with shared variables.
Namely, though for each conjunct $c = \basis_3(p_s,p_r,p_{s,r})$ we can break down an arbitrary assignment of variables in $c$ into a rank-$1$ assignment, there are many inconsistent ways of collecting all these assignments to produce a rank-$1$-assignment for the entire conjunction.
In order to overcome this issue, we must acknowledge the inherent contextuality in the $\strongSAT{\Pj(\RAThree)}$ problem.
Thus, as part of our reduction, we use the orthogonality graph $G_{\CKS}$ of the basis-complete $40$-vector Kochen-Specker set $\CKS$ mentioned in section~\ref{sec:kochen-specker} as a gadget.
The primary property of $G_{\CKS}$ we use is that any orthogonal assignment is guaranteed to have a triangle which is assigned to only rank-$1$ projectors.
\begin{proposition}
  \label{prop:peres-choice}
  If $f\colon G_{\CKS} \rightarrow \Pj(\RAThree)$ is an orthogonal assignment, then there exist orthonormal bases $\{\vec{x},\vec{y},\vec{z}\} \in \mathcal{B}_{\CKS}$ and $\{\vec{u}_x,\vec{u}_y,\vec{u}_z\}$ such that
  $f(\vec{x}) = E_{\vec{u_x}}$, $f(\vec{y}) = E_{\vec{u_y}}$, and $f(\vec{z}) = E_{\vec{u_z}}$ are rank-$1$ projectors.
\end{proposition}
\begin{proof}
  Consider the function $\rank \circ f\colon \CKS \rightarrow \{0,1,2,3\}$.
  The function $\rank \circ f$ induces a `colouring' function $g\colon \CKS \rightarrow \{0,1\}$ where $g(\vec{v}) = \rank(f(\vec{v})) \bmod 2$.
  From $f$ being an orthogonal assignment, for every basis $\{\vec{x'},\vec{y'},\vec{z'}\} \in \mathcal{B}_{\CKS}$,
  \begin{itemize}
    \item $\rank(f(\vec{x'})) + \rank(f(\vec{y'})) + \rank(f(\vec{z'})) = 3$.
    \item $g(\vec{x'}) + g(\vec{y'}) + g(\vec{z'}) = 1$ or $g(\vec{x'}) + g(\vec{y'}) + g(\vec{z'}) = 3$.
  \end{itemize}
  Since $\CKS$ is a complete $3$-dimensional Kochen-Specker proof, by Proposition~\ref{prop:bc-colouring}, there exists a maximum clique $C_{g}$ of $G_{S}$, or equivalently a basis $B_{W} = \{\vec{x},\vec{y},\vec{z}\} \in \mathcal{B}_{\CKS}$, such that $g(\vec{x}) + g(\vec{y}) + g(\vec{z}) \not= 1$.
  Thus, $g(\vec{x}) + g(\vec{y}) + g(\vec{z}) = 3$.
  From this constraint and $\rank(f(\vec{x})) + \rank(f({\vec{y}})) + \rank(f(\vec{z})) = 3$, we can deduce that $\rank(f({\vec{x}})) = \rank(f(\vec{y})) = \rank(f({\vec{z}})) = 1$.
  Therefore, $E_{\vec{u}_x} = f(\vec{x})$, $E_{\vec{u}_y} = f({\vec{y}})$, and $E_{\vec{u}_z} = f({\vec{z}})$ are rank-$1$ projectors for some orthonormal basis $B_{U} = (\vec{u_x},\vec{u_y},\vec{u_z})$ of $\mathbb{R}^3$.
\end{proof}

We use two copies of $K_3$ to form the `consistency' gadget $\Cons$ depicted in Figure~\ref{fig:cons} to ensure that different vertices are assigned the same projector by any orthogonal assignment.
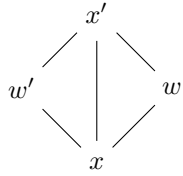
\begin{figure}[h]
  \caption{The gadget $(\Cons,w,w')$}
  \centering
  \label{fig:cons}
  \begin{tikzpicture}
    \begin{scope}[every node/.style={circle}]

    \node (x) at (1,-1) {$x$};
    \node (x') at (1,1) {$x'$};
    \node (w') at (0,0) {$w'$};
    \node (w) at (2,0) {$w$};
    \end{scope}

    \begin{scope}
      \path (x) edge (x');
      \path (x) edge (w);
      \path (x') edge (w);
      \path (x') edge (w');
      \path (x) edge (w');
    \end{scope}
  \end{tikzpicture}
\end{figure}
The following proposition expresses the how $\Cons$ is used in the reduction.
\begin{proposition}
  \label{prop:cons}
    If $f\colon \Cons \rightarrow \Pj(\RAThree)$ is an orthogonal assignment such that $f(w) = E_{\vec{w}}$, $f(w') = E_{\vec{w'}}$ are rank-$1$ projectors corresponding to unit vectors $\vec{w},\vec{w'} \in \RAThree$, then $f(w') = f(w) = E_{\vec{w}}$.
\end{proposition}
\begin{proof}
Let $Q = \Pj(\RAThree)$. Follows straightforwardly from $f$ being an orthogonal assignment:
\begin{align*}
  f(w) &= \neg_Q(\neg_Q f(w)) \\
       &= \neg_Q(f(x) \vee_{Q} f(x')) \\
       &= f(w')
\end{align*}
\end{proof}

We now can convert every cross-product term $t(x_1,\dots,x_n)$ into a graph $G_t$.
The construction of $G_t$ is through the following steps:
\begin{enumerate}
  \item View $t(x_1,\dots,x_n)$ as its rooted directed parse tree $T_t$.
  \item Produce a directed-acyclic graph $D_t$ by identifying all leaf nodes in $T_t$ that are occurrences of the same variable $x_i$.
  \item \label{item:root-var-1} Produce a directed graph $D'_t$ by identifying the root node $t$ of $D_t$ with the sink node representing variable $x_1$.
  \item Produce an undirected graph $G'_t$ by forgetting the orientations of edges in $D'_t$.
  \item Produce $G_t$ by taking the triangle completion of $G'_t$, i.e.\ for every edge $e = (u,v)$ add an additional vertex $v_e$ incident to both $u$ and $v$ ensuring $\{u,v,v_e\}$ is a triangle.
\end{enumerate}

For every graph $G$, we construct $\CKER(G)$ through the following steps:
\begin{enumerate}
  \item Take the categorical product $G \times G_{\CKS}$. Recall that the categorical product $G \times H$ of two graphs $G$ and $H$ has vertex set $V(G \times H) = V(G) \times V(H)$ and edge set defined as:
\[ E(G \times H) = \{((v,v'),(w,w')) \mid (v,w) \in E(G) \text{ and } (v',w') \in E(H) \} \]
\item \label{item:synthetic-equality} To construct $\CKER(G)$, for every $w \in V(G)$ and pair of distinct vectors $\vec{v},\vec{u} \in \CKS$ we add a consistency gadget $(\Cons,(w,\vec{v}),(w,\vec{u}))$ to $G \times G_{\CKS}$.
\end{enumerate}
The primary property of the $\CKER(\cdot)$ construction we use is stated in the following lemma.
\begin{figure}[h]
  \label{fig:pe-cons}
\end{figure}
\begin{lemma}
  \label{lem:force-rank-1}
  Let $G$ be a $3$-facet graph.
  For every orthogonal assignment $f\colon \CKER(G) \rightarrow \Pj(\FK^3)$, there exists a rank-$1$ orthogonal assignment $\nu_{f}\colon G \rightarrow \Pj_1(\FK^3)$.
\end{lemma}
\begin{proof}
  For every triangle $C = \{a,b,c\} \in \Omega(G)$ of $G$, there is subgraph $G^{C}_{\CKS}$ in $\CKER(G)$ isomorphic to $G_{\CKS}$ with vertices:
  \[ V(G^{C}_{\CKS}) = \bigcup_{\{\vec{x},\vec{y},\vec{z}\} \in \mathcal{B}_{\CKS}} \{(a,\vec{x}),(b,\vec{y}),(c,\vec{z})\} \]
  By Proposition~\ref{prop:peres-choice}, there exists orthonormal bases $\{\vec{u}_a,\vec{u}_b,\vec{u}_c\} \in \mathcal{B}_{\CKS}$ and $\{\vec{v}_a,\vec{v}_b,\vec{v}_c\} \subseteq \RAThree$ such that
  $f((e,\vec{u}_e)) = E_{\vec{v}_e}$ is a rank-$1$ projector for all $e \in C$.
  Thus, for every $C = \{a,b,c\} \in \Omega(G)$, there is
  \begin{itemize}
    \item a triangle $H(C) = \{(a,\vec{u}_a),(b,\vec{u}_b),(c,\vec{u}_c)\} \in \Omega(\CKER(G))$ with isomorphism $\mu_{C}\colon C \rightarrow H(C)$, and
    \item a rank-$1$ orthogonal assignment $\nu_{C}\colon C \rightarrow \Pj_1(\RAThree)$ defined as $\res{f}{H(C)} \circ \mu_{C}$, i.e.\ $\nu_{C}(e) = f(e,\vec{u}_e) = E_{\vec{v}_e}$ for all $e \in C$.
  \end{itemize}
  We now claim the family of local sections $\{\nu_{C}\}_{C \in \Omega(G_t)}$ agrees on intersections, i.e.\ if $w \in C \cap C'$, then $\nu_{C}(w) = \nu_{C'}(w)$.
  By step~\ref{item:synthetic-equality}, the elements $\mu_{C}(w) \in H(C)$ and $\mu_{C'}(w) \in H(C')$ are connected a consistency gadget $(\Cons,\mu_{C}(w),\mu_{C'}(w))$.
  Moreover, $f(\mu_{C}(w))$ and $f(\mu_{C'}(w))$ are rank-$1$ projectors.
  By Proposition~\ref{prop:cons}, $\nu_{C}(w) = f(\mu_{C}(w)) = f(\mu_{C'}(w)) = \nu_{C'}(w)$.
  Since the family of local sections $\{\nu_{C}\}_{C \in \Omega(G_t)}$ agrees on intersections, there exists a global section $\nu_{f}\colon G \rightarrow \Pj_1(\RAThree)$ with definition $\nu_{f}(w) = \nu_{C}(w)$ for any $w \in C \in \Omega(G)$.
  The function $\nu_{f} \colon G \rightarrow \Pj_1(\RAThree)$ is a rank-$1$ orthogonal assignment because each of the local sections $\nu_{C}$ is a rank-$1$ orthogonal assignment.
\end{proof}

Finally, we use the propositional formula $\theta_t(\vec{p}) = \varphi_{\CKER(G_t)}(\vec{p})$, i.e.\ the orthogonal assignment formula of graph $\CKER(G_t)$ to map an instance of $\XSAT(\RPTwo)$ to $\SAT{\Pj(\RAThree)}$.

\begin{proposition}
  \label{prop:xsat-reduction-r3}
  For every cross-product term $t$, the following are equivalent
  \begin{enumerate}[label=(\arabic*)]
    \item \label{item:xsat-reduction-r3-cross} $t(x_1,\dots,x_n) \in \XSAT(\RPTwo)$.
    \item \label{item:xsat-reduction-r3-graph} There exists an orthogonal assignment $f\colon \CKER(G_t) \rightarrow \Pj(\mathbb{R}^3)$.
    \item \label{item:xsat-reduction-r3-strongr} $\theta_{t}(\vec{p}) \in \strongSAT{\Pj(\RAThree)}$.
    \item \label{item:xsat-reduction-r3-weakr} $\theta_{t}(\vec{p}) \in \weakSAT{\Pj(\RAThree)}$.
    \end{enumerate}
\end{proposition}
\begin{proof}
  Throughout the proof, we use the a bijection $b\colon \RPTwo \rightarrow \Pj_1(\RAThree)$ which maps a line $\mathbb{R}\vec{v}$ to its rank-$1$ projector $E_{\vec{v}}$.

  For the $\ref{item:xsat-reduction-r3-cross} \Rightarrow \ref{item:xsat-reduction-r3-graph}$ implication, suppose $t(x_1,\dots,x_n)$ is satisfied by the lines $\mathbb{R}\vec{v_1},\dots,\mathbb{R}\vec{v_n} \in \RPTwo$.
  This allows us to produce a rank-$1$ orthogonal assignment $g'\colon G'_t \rightarrow \Pj_1(\mathbb{R}^3)$ where $g(s) = b(\mathbb{R}\langle s \rangle)$ for a subterm $s$ of $t$.
  We can extend $g'\colon G'_t \rightarrow \Pj_1(\RAThree)$ to a rank-$1$ orthogonal assignment $g\colon G_t \rightarrow \Pj_1(\RAThree)$ for the triangle completion $G_t$ of $G'_t$ by letting $g(v_{(s,r)}) = \ocomp{g'(s) + g'(r)}$, i.e.\ projector onto the line orthogonal to the plane spanned by the lines $\mathbb{R}\langle s \rangle$ and $\mathbb{R} \langle r \rangle$ and $g(s) = g'(s)$ otherwise.
  Next, we extend the rank-$1$ orthogonal assignment $g\colon G_t \rightarrow \Pj_1(\RAThree)$ to a rank-$1$ orthogonal assignment $h\colon G_t \times G_{\CKS} \rightarrow \Pj_1(\RAThree)$ where $h((e,\vec{u})) = g(e)$.
  The final step is to extend the rank-$1$ orthogonal assignment $h\colon G_t \times G_{\CKS} \rightarrow \Pj_1(\RAThree)$ to an orthogonal assignment $f\colon \CKER(G_t) \rightarrow \Pj_1(\RAThree)$.
  For every $w$ and distinct vectors $\vec{u},\vec{v}$, there is a consistency gadget $(\Cons,(w,\vec{u}),(w,\vec{v}))$ in $\CKER(G_t)$ with two additional vertices $x,x'$.
  By construction, $h(w,\vec{u}) = h(w,\vec{v}) = g(w) = E_{\vec{w}}$, so we can set $f(x)$ and $f(x')$ in the gadget $(\Cons,(w,\vec{u}),(w,\vec{v}))$ to basis vectors of the plane orthogonal to $\vec{w}$.
  We let $f$ coincide with $h$ on the subgraph $G_t \times G_{\CKS}$ of $\CKER(G_t)$ which does not contain consistency gadgets.
  By construction, $f$ is an orthogonal assignment.

  For the $\ref{item:xsat-reduction-r3-cross} \Leftarrow \ref{item:xsat-reduction-r3-graph}$ implication, by Lemma~\ref{lem:force-rank-1}, there exists a rank-$1$ orthogonal assignment $\nu_{f}\colon G_t \rightarrow \Pj_1(\RAThree)$.
  For all $i \in \setn$, the lines $\mathbb{R}\vec{v_i} = b^{-1}(\nu_{f}([x_i]))$ in $\RPTwo$ define a satisfying instance of $t(x_1,\dots,x_n)$.

  The $\ref{item:xsat-reduction-r3-graph} \Leftrightarrow \ref{item:xsat-reduction-r3-strongr}$ bi-implication is an application of Proposition~\ref{prop:ortho-assignments}.

  For the $\ref{item:xsat-reduction-r3-strongr} \Leftrightarrow \ref{item:xsat-reduction-r3-weakr}$ bi-implication, the $\Rightarrow$ direction is obvious since $1_{Q} \not= 0_{Q}$ for $Q = \Pj(\mathbb{R}^3)$.
  Conversely, suppose for contradiction $\theta_t(\vec{p})$ is weakly satisfied, but not strongly satisfied.
  By Proposition~\ref{prop:ortho-assignments}, this induces an orthogonal assignment $f\colon \CKER(G_t) \rightarrow \Pj(\mathbb{R}^c)$ for some $c \in \{1,2\}$.
  The function $f$, by restriction to a copy of $G_{\CKS}$ in $\CKER(G_t)$ and composing with a $\pBA$ morphism $\Pj(\mathbb{R}^c) \rightarrow \Two$ (which exists by $c < 3$), induces a non-contextual colouring $g\colon G_{\CKS} \rightarrow \{0,1\}$.
  Contradiction, since $G_{\CKS}$ is Kochen-Specker proof, there can be no non-contextual colouring.
\end{proof}

The reduction exhibited in Proposition~\ref{prop:xsat-reduction-r3} establishes that $\strongSAT{\Pj(\RAThree)}$ and $\weakSAT{\Pj(\RAThree)}$ are $\ER$-hard.
In order to generalise these hardness results to dimension $d > 3$, we need a padding argument to obtain a reduction from $\XSAT(\RPTwo)$ to the satisfiability problems over $\Pj(\mathbb{R}^d)$.
This padding argument is encapsulated in Proposition~\ref{prop:dimension-padding} below.

For every propositional formula $\phi(p_1,\dots,p_n)$ and dimension $d$, we define a \emph{$d$-dimensional realisation} formula $\hat{\phi}_d(p_1,\dots,p_n,q_1,\dots,q_{d+1})$ defined as the conjunction of:
\begin{enumerate}[label=(D\arabic*)]
    \item \label{item:drealise-basis} $\basis_{d+1}(q_1,\dots,q_{d+1})$
    \item \label{item:drealise-subspace} $\bigwedge_{i \in \setn} \phi_{\leq}(p_i,\neg q_{d+1})$
    \item \label{item:drealise-onto} $\phi(p_1,\dots,p_n) \leftrightarrow \neg q_{d+1}$
\end{enumerate}
The purpose of this formula is to show that we can embed the witnesses to satisfiability of $\phi(p_1,\dots,p_n)$ in $\FK^d$ into a $d$-dimensional subspace of $\mathbb{K}^{d+1}$.
Suppose $\gamma \in \mdom{Q}{\hat{\phi}}$ is a meaningful substitution.
Item~\ref{item:drealise-basis}, asserts that if the projectors $\gamma(q_i)$ assigned to the $q$ variables are all non-zero, then they form a PVM of rank-$1$ projectors, i.e.\ pick out an orthonormal basis of $\FK^{d+1}$.
Item~\ref{item:drealise-subspace} asserts that each of the projectors $\gamma(p_j)$ assigned to the $p$-variables is a projector onto some subspace of the $d$-dimensional subspace $S$ orthogonal to the $1$-dimensional subspace $\gamma(q_{d+1})$, i.e.\ the image of $\sum_{i \in \setd} \gamma(q_d)$.
Item~\ref{item:drealise-onto} asserts that $\phi(p_1,\dots,p_n)$ is assigned to the projector on $S$.
\begin{proposition}
  \label{prop:dimension-padding}
  Let $\vec{p} = (p_1,\dots,p_n)$ and $\vec{q} = (q_1,\dots,q_{d+1})$.
  $\phi(\vec{p}) \in \strongSAT{\Pj(\FK^d)}$ if, and only if, $\hat{\phi}_d(\vec{p},\vec{q}) \in \strongSAT{\Pj(\FK^{d+1})}$.
\end{proposition}
\begin{proof}
$\Rightarrow$ Fix an orthonormal basis $e_1,\dots,e_{d+1}$ of $\Hilb_{d+1}$ and let $E_1,\dots,E_{d+1}$ be their corresponding rank-$1$ projectors.
Let $S$ be the subspace of $\Hilb_{d+1}$ of dimension $d$ spanned by $e_1,\dots,e_{d}$.
By standard linear algebra, there exists an isomorphism $t\colon \Hilb_{d} \rightarrow S$ induced by a basis $f_1,\dots,f_{d}$ and such that $f_i \mapsto e_i$.
If $\Pj(\Hilb_d)$ satisfies $\varphi(p_1,\dots,p_n)$ using projectors $P_i$.
Express $P_i$ as a linear combination of projectors $F_1,\dots,F_d$ onto basis vectors $f_1,\dots,f_{d}$.
Using the the $t$ isomorphism, we obtain projectors onto subspaces of $S$.
Thus, assigning $p_i$ to the image of $P_i$ under $t$ and $q_z$ to $E_z$, we obtain that $\Pj(\Hilb_{d+1})$ satisfies $\hat{\phi}_d$.

$\Leftarrow$ Conversely, if $\Pj(\Hilb_{d+1})$ satisfies $\hat{\phi}_d$, then we can choose a satisfying assignment such that \ref{item:drealise-basis} is satisfied by the rank-$1$ projectors associated to an orthonormal basis $e_1,\dots e_{d+1}$ of $\Hilb_{d+1}$.
Item \ref{item:drealise-subspace} allows us to view each $p_i$ as assigned to projector onto a subspace $S_i$ of the subspace $S$ spanned by $e_1,\dots,e_d$, i.e.\ the orthogonal complement of the line spanned by $e_{d+1}$.
From here, the proof is the similar to the converse, but using the isomorphism $u\colon S \rightarrow \Hilb_{d}$ and item~\ref{item:drealise-onto}, to ensure that $\Pj(\Hilb_d)$ satisfies $\phi(p_1,\dots,p_n)$.
\end{proof}

\begin{theorem}
\label{thm:fixed-d-real-sat}
For every $d \geq 3$, $\weakSAT{\Pj(\mathbb{R}^d)}$ and $\strongSAT{\Pj(\mathbb{R}^d)}$ are $\ER$-complete.
\end{theorem}
\begin{proof}
  Proposition~\ref{prop:realdsat-er} proves membership in $\ER$.
  For $\ER$-hardness of the base case $d = 3$, Proposition~\ref{prop:xsat-reduction-r3} yields reductions from the $\ER$-complete problem $\XSAT(\RPTwo)$ to $\weakSAT{\Pj(\mathbb{R}^3)}$ and $\strongSAT{\Pj(\mathbb{R}^3)}$.
  The size of $\CKER(G_t)$ is $O(|G_{\CKS}||G_t|) = O(40|t|)$ and thus $\theta_t(\vec{p})$ is polynomial in the size of $t(x_1,\dots,x_n)$.
  For $\ER$-hardness beyond the base case, we inductively apply Proposition~\ref{prop:dimension-padding} to obtain reduction from $\strongSAT{\Pj(\mathbb{R}^d)}$ to $\strongSAT{\Pj(\mathbb{R}^{d+1})}$.
  Finally, the $\weakSAT{\Pj(\mathbb{R}^{d})}$ problem is equivalent to $\strongSAT{\Pj(\mathbb{R}^{d'})}$ for some $d' \leq d$.
  The $d$-dimensional padding formula $\hat{\phi}_{d}(\vec{p},\vec{q})$ from Proposition~\ref{prop:dimension-padding} is linear in the size of $\phi(\vec{p})$.
\end{proof}

\subsection{Complex case}
We define the class $\exists \mathbb{C}$ as $\exists (L,M)$ where $L = \Lrings$ and $M$ is the field of complex numbers $(\mathbb{C},0,1,+,*)$.
It is easy to see that $\exists \mathbb{C} \subseteq \exists \mathbb{R}$ since there is an isomorphism of rings mapping a complex number $z = x + iy$ to the $2 \times 2$-real matrix $\left( {\begin{smallmatrix} x & y \\ -y & x \\ \end{smallmatrix}} \right)$.
We also define $\ECconj$ to be $\exists (L,M)$ where $L = \Lrings \cup \{u\}$ has an additional unary function symbol $u$ and $M$ is the complex field $(\mathbb{C},0,1,+,*,(\cdot)^{*})$ equipped with complex conjugation $(\cdot)^{*}$ as the interpretation of $u$.
As the subfield $\mathbb{R}$ can be defined as the numbers fixed by complex conjugation, we obtain the equality between the complexity classes $\ECconj = \ER$
\begin{proposition}
\label{prop:ec-to-er}
$\EC \subseteq \ER$ and $\ECconj = \ER$.
\end{proposition}
To show that satisfiablity in $\Pj(\mathbb{C}^d)$ is in $\ER$, we use a similar proof to Proposition~\ref{prop:realdsat-er} to produce a sentence in $\ECconj$ from an input propositional formula.
The result then follows from Proposition~\ref{prop:ec-to-er}.
\begin{proposition}
  \label{prop:complexdsat-er}
  For every $d \geq 1$, $\strongSAT{\Pj(\mathbb{C}^d)}$ and $\weakSAT{\Pj(\mathbb{C}^d)}$ are in $\ER$.
\end{proposition}
\begin{proof}
Similarly to the proof of Proposition~\ref{prop:realdsat-er}, we first translate propositional satisfiability in $\Pj(\mathbb{C}^d)$ into the existence of a self-adjoint solution set of matrix equations in $M_{\mathbb{C}}(d,d)$.
We then translate the matrix equations in $M_{\mathbb{C}}(d,d)$ and the self-adjoint condition into equations over $\mathbb{C}$.
This process is via unfolding the definitions of matrix addition and multiplication in terms of matrix entries.
The salient difference is that in expressing the self-adjoint condition, we now use complex conjugation, i.e. $P_{ij} = u(P_{ji})$.
Therefore, $\strongSAT{\Pj(\mathbb{C}^d)}$ and $\weakSAT{\Pj(\mathbb{C}^d)}$ are in $\ECconj$, and by Proposition~\ref{prop:ec-to-er} in $\ER$.
\end{proof}

We now proceed to prove that  $\strongSAT{\Pj(\mathbb{C}^d)}$ for $d \geq 4$ is $\ER$-hard, and thus $\ER$-complete by Proposition~\ref{prop:complexdsat-er}.
The reduction from $\XSAT(\RPTwo)$ we employed in the real case (Proposition~\ref{prop:xsat-reduction-r3}) does not straightforwardly translate to the complex case.
This is because we need to produce orthogonal sets of real lines in $\mathbb{R}^3$ from a meaningful substitution which consists of complex projectors.
That is we need to translate data in the form of operators over $\mathbb{C}^d$ to lines in $\mathbb{R}^3$.
In order to accomplish this, we have to change the base case of our dimension from $d = 3$ to $d = 4$.
However, the advantage of this reduction is that we also obtain a reduction from $\XSAT(\RPTwo)$ to the exclusive-or satisfiability problem $\strongXOR{\Pj(\FC^d)}$ that we describe below.

The first ingredient in our reduction from $\XSAT(\RPTwo)$ to $\strongSAT{\Pj(\mathbb{C}^4)}$ is the Pauli vector isomorphism from $\mathbb{R}^3$ to the set $\Bnd_0(\FC^2)$ of traceless Hermitian matrices over $\mathbb{C}^2$.
This isomorphism is the map $\vec{\sigma}\colon \mathbb{R}^3 \rightarrow \Bnd_0(\FC^2)$ given by $(a_1,a_2,a_3) \mapsto (a_1 \sigma_x + a_2 \sigma_y + a_3 \sigma_z)$ where $\sigma_x,\sigma_y,\sigma_z$ are the standard Pauli matrices.
More explicitly,
\[ \vec{\sigma}(a_1,a_2,a_3) =
  \begin{pmatrix}
    a_3 & a_1 - \iu a_2 \\
    a_1 + \iu a_2 & -a_3
  \end{pmatrix}
\]
The second ingredient in our reduction is to change perspective by using the pBA $\Inv(\FC^d)$ consisting of self-adjoint involutions, i.e.\ $A^2 = I_d$ instead of $\Pj(\FC^d)$.
For every $d \in \mathbb{N}$, there is a bijection $b\colon \Pj(\mathbb{C}^d) \rightarrow \Inv(\mathbb{C}^d)$ given by $b(E) = I_d - 2E$ which preserves commutative pairs.
This bijection allows us to reverse engineer the definitions of $\vee,\wedge$ for $\Inv(\mathbb{C}^d)$ from their definitions in the pBA $\Pj(\mathbb{C}^d)$.
However, our main reason for shifting to $\Inv(\mathbb{C}^d)$ is that the exclusive-or operation $\xor$ on involutions $\Inv(\FC^d)$ is matrix multiplication.
Moreover, $\Inv(\FC^2) \subseteq \Bnd_0(\FC^2)$ and contains the standard Pauli matrices.
Note the in the classical case of $d = 1$, this amounts to moving the Boolean domain from $\{0,1\}$ to $\{+1,-1\}$.

The final ingredient in our reduction is the so-called Peres-Mermin Magic Square~\cite{peresKS}.
The Peres-Mermin magic square is a witness to Kochen-Specker contextuality through quantum entanglement.
The magic square $M$ is a set of equations over the group $J_2 = (\{-1,+1\},*) \cong (\mathbb{Z}_{2},+)$ typically presented as the following table:
\begin{equation}
\label{eq:magic-square-variables}
\begin{array}{|c|c|c:}
\hline
a & b & c \\ \hline
d & e & f\\ \hline
g & h & i \\ \hline
\end{array}
\end{equation}
where each cell represents a variable, every row and the first 2 columns represent a multiplicative equation equal to $1$ (e.g.\ $abc = 1$), and the last column represents the equation $cfi = -1$ (indicated by the dotted column seperator).
However, these equations are satisfiable in $\Inv(\mathbb{C}^4)$ via the assignment expressed in the table using Pauli matrices acting on two qubits:
\begin{equation}
\label{eq:magic-square-standard-assignment}
\begin{array}{|c|c|c:}
\hline
\sigma_x \otimes I_2 & I_2 \otimes \sigma_x & \sigma_x \otimes \sigma_x \\ \hline
I_2 \otimes \sigma_Z & \sigma_Z \otimes I_2 & \sigma_x \otimes \sigma_z \\ \hline
\sigma_x \otimes \sigma_z & \sigma_z \otimes \sigma_x & \sigma_y \otimes \sigma_y \\ \hline
\end{array}
\end{equation}
These yield a satisfying assignment since every row and the first 2 columns multiply to $I_{4}$ and the entries of the last column multiply to $-I_{4}$.
The magic square equations can be written as a propositional formula $\mu(\vec{a})$ involving the variables $[\vec{a}]$ listed in the table~\eqref{eq:magic-square-variables} and the exclusive or operation $\xor$.
Namely, with the exception of the last column, every row and column corresponds to a positive XOR-clause of the variables, e.g.\ the first row is $a \xor b \xor c$.
The last column corresponds as negative XOR-clause $\neg (c \xor f \xor i)$.
Since $M$ does not have solution over $\FZ_{2}$, we have that $\mu(\vec{a})$ is unsatisfiable in the two-element Boolean algebra $\Two$.
The table~\eqref{eq:magic-square-standard-assignment} demonstrates that $\mu(\vec{a})$ is strongly satisfiable in $\Inv(\mathbb{C}^4)$.
We call such strongly satisfying assignments $\alpha \in Q^{\mu(\vec{a})}$ \emph{solutions to the magic square $M$}.

In Lemma~\ref{lem:cross-to-magic}, we show that we can express the orthogonality of lines in $\RPTwo$ using solutions to the magic square $M$.
This is because an orthogonal set of lines in $\RAThree$ corresponds to a qubit basis in $\Bnd_0(\mathbb{C}^2)$ under the bijection $\vec{\sigma}$.
A set of involutions $\{X,Y,Z\}$ is a \emph{qubit basis} if the following equations are satisfied:
\begin{align*}
  [X,Y] = 2iZ &\quad \{X,Y\} = 0 \\
  [Z,X] = 2iY &\quad \{Z,X\} = 0 \\
  [Y,Z] = 2iX &\quad \{Y,Z\} = 0
\end{align*}
where $[X,Y] := XY - YX$ is the commutator operation and $\{X,Y\} := XY + YX$ is the anti-commutator operation for all $X,Y \in \Bnd(\Hilb)$.
\begin{proposition}
  \label{prop:involution-facts}
  Let $A,B,X,Y \in I(\Hilb)$ be involutions over $\FK$ Hilbert space $\Hilb$.
  \begin{enumerate}[label=(\arabic*)]
    \item \label{item:unit-commute} $(AB)^2 = I_{\Hilb}$ iff $AB = BA$, i.e.\ $A$ and $B$ commute.
    \item \label{item:unit-anti-commute} $(AB)^2 = -I_{\Hilb}$ iff $AB = -BA$, i.e.\ $A$ and $B$ anti-commute.
    \item \label{item:tensor-factor} If $A \otimes B = X \otimes Y$, then $A = \pm X$ and $B = \pm Y$.
  \end{enumerate}
\end{proposition}
\begin{proof}
  For \ref{item:unit-commute} and \ref{item:unit-anti-commute}, we have
  \[ (AB)^2 = ABAB = A(BA)B. \]
  In the case where they commute $AB = BA$, the above equation reduces $A^2B^2$ which is $I_{\Hilb}$ by $A,B$ being involutions.
  In the case where they anti-commute, $AB = -BA$, the above equation reduces to $-A^2B^2$ and thus $-I_{\Hilb}$.
  The converse for both cases follow from simple calculation.

  For \ref{item:tensor-factor}, recall that for general nonzero operators, $A \otimes B = X \otimes Y$ implies that
  $A = \lambda X$ and $B = \frac{1}{\lambda} Y$ for some nonzero $\lambda \in \FK$.
  Since all operators involved are involutions, $A^2 = \lambda^2 X = I$ implies $\lambda^2 = 1$.
  Thus, $\lambda = \frac{1}{\lambda} = \pm 1$.
\end{proof}

\begin{lemma}
  \label{lem:cross-to-magic}
  Suppose $\vec{x},\vec{y},\vec{z}$ are unit vectors in $\RAThree$, $d \geq 1$, and $Q = \Inv(\mathbb{C}^{4d})$ then the following are equivalent:
  \begin{enumerate}[label=(\arabic*)]
    \item \label{item:c2m-cross} $\{\vec{x},\vec{y},\vec{z}\}$ is an orthonormal basis of $\RAThree$.
    \item \label{item:c2m-commutator} $\{\vec{\sigma}(\vec{x}),\vec{\sigma}(\vec{y}),\vec{\sigma}(\vec{z})\}$ is a qubit basis.
    \item \label{item:c2m-magic} There exists solution $\alpha \in Q^{\mu(\vec{a})}$ to the magic square expressed as:
  \[
  \begin{array}{|c|c|c:}
    \hline
    \vec{\sigma}(\vec{x}) \otimes I \otimes I_d & I \otimes \vec{\sigma}(\vec{x}) \otimes I_d & \vec{\sigma}(\vec{x}) \otimes \vec{\sigma}(\vec{x}) \otimes I_d\\ \hline
    I \otimes \vec{\sigma}(\vec{z}) \otimes I_d & \vec{\sigma}(\vec{z}) \otimes I \otimes I_d & \vec{\sigma}(\vec{z}) \otimes \vec{\sigma}(\vec{z}) \otimes I_d\\ \hline
    \vec{\sigma}(\vec{x}) \otimes \vec{\sigma}(\vec{z}) \otimes I_d & \vec{\sigma}(\vec{z}) \otimes \vec{\sigma}(\vec{x}) \otimes I_d & \vec{\sigma}(\vec{y}) \otimes \vec{\sigma}(\vec{y}) \otimes I_d\\ \hline
    \end{array}
    \]
  \end{enumerate}
\end{lemma}
\begin{proof}
  For the $\ref{item:c2m-cross} \Leftrightarrow \ref{item:c2m-commutator}$ implication, consider $\vec{u} = (u_1,u_2,u_3), \vec{v} = (v_1,v_2,v_3) \in \RAThree$.
  Applying the map $\vec{\sigma}$ to these vectors we obtain:
    \[
      \vec{\sigma}(\vec{u}) =
    \begin{pmatrix}
      u_3 & u_1 - i u_2\\
      u_1 + i u_2 & -u_3
    \end{pmatrix}
    , \quad
    \vec{\sigma}(\vec{v}) =
    \begin{pmatrix}
      v_3 & v_1 - i v_2\\
      v_1 + i v_2 & -v_3
    \end{pmatrix}
    \]
Multiplying in one order yields
    \[
      \vec{\sigma}(\vec{u})\vec{\sigma}(\vec{v}) =
      \begin{pmatrix}
        u_3 v_3 + (u_1 - i u_2)(v_1 + i v_2)
        &
        u_3 (v_1 - i v_2) - (u_1 - i u_2) v_3
        \\[6pt]
        (u_1 + i u_2) v_3 - u_3 (v_1 + i v_2)
        &
        (u_1 + i u_2)(v_1 - i v_2) + u_3 v_3
      \end{pmatrix}.
    \]
and the other order yields
    \[
      \vec{\sigma}(\vec{v})\vec{\sigma}(\vec{u}) =
      \begin{pmatrix}
      v_3 u_3 + (v_1 - i v_2)(u_1 + i u_2)
      &
      v_3 (u_1 - i u_2) - (v_1 - i v_2) u_3
      \\[6pt]
      (v_1 + i v_2) u_3 - v_3 (u_1 + i u_2)
      &
      (v_1 + i v_2)(u_1 - i u_2) + v_3 u_3
      \end{pmatrix}.
    \]
allowing us to compute the commutator as:
    \begin{align}
      [\vec{\sigma}(\vec{u}), \vec{\sigma}(\vec{v})]
      &= 2i
      \begin{pmatrix}
        (u_1 v_2 - u_2 v_1) &
        (u_2 v_3 - u_3 v_2) - i(u_3 v_1 - u_1 v_3) \\[6pt]
        (u_2 v_3 - u_3 v_2) + i(u_3 v_1 - u_1 v_3) &
        -(u_1 v_2 - u_2 v_1)
      \end{pmatrix} \notag \\
      &= 2i\vec{\sigma}(\vec{u} \cross \vec{v}) \label{eq:pauli-commutator}
    \end{align}
and the anti-commutator as:
    \[
      \{\vec{\sigma}(\vec{u}), \vec{\sigma}(\vec{v})\}
      = 2
      \begin{pmatrix}
        u_3v_3 + u_1v_1 + u_2v_2 &
        0 \\[6pt]
        0 &
        u_3v_3 + u_1v_1 + u_2v_2
      \end{pmatrix}
      = 2(\vec{u} \cdot \vec{v})I
    \]
    Thus, the cross and dot product of the vectors are transferred by $\vec{\sigma}$ to the commutator and anti-commutator operations.
    From the hypotheses of \ref{item:c2m-cross} and substituting unit vectors $\vec{x},\vec{y},\vec{z}$ for $\vec{u}$ or $\vec{v}$ above, we can conclude the equations for a qubit basis are satisfied by $\{\vec{\sigma}(\vec{x}),\vec{\sigma}(\vec{y}),\vec{\sigma}(\vec{z})\}$.
    Conversely, from the fact that $\{\vec{\sigma}(\vec{x}),\vec{\sigma}(\vec{y}),\vec{\sigma}(\vec{z})\}$ is a qubit basis, we can use equation~\eqref{eq:pauli-commutator} to conclude that the vectors $\{\vec{x},\vec{y},\vec{z}\}$ are orthogonal.

    For the $\ref{item:c2m-commutator} \Rightarrow \ref{item:c2m-magic}$ implication,
    For every $X,Y \in \Bnd_0(\mathbb{C}^2)$, we have the equations
    \[ XY = \frac{1}{2}(\{X,Y\} + [X,Y]), \]
    \[ YX = \frac{1}{2}(\{X,Y\} - [X,Y]). \]
    We also have the mixed-product rule stating
    \[ (X_1 \otimes Y_1)(X_2 \otimes Y_2) = (X_1 X_2) \otimes (Y_1 Y_2). \]
    Applying these equations and the qubit basis equations for the sets $\{\vec{\sigma}(\vec{x_1}),\vec{\sigma}(\vec{y_1}),\vec{\sigma}(\vec{z})\}$ and $\{\vec{\sigma}(\vec{x_2}),\vec{\sigma}(\vec{y_2}),\vec{\sigma}(\vec{z})\}$, allow us to verify that the entries in the table reflect the magic square equations.

    For the $\ref{item:c2m-commutator} \Leftarrow \ref{item:c2m-magic}$ implication, we ignore the extra $d$ tensor factor.
    Consider the last column equation:
    \begin{align*}
      -I_2 \otimes I_2 &= (\vec{\sigma}(\vec{x_1}) \otimes \vec{\sigma}(\vec{x_2}))(\vec{\sigma}(\vec{z}) \otimes \vec{\sigma}(\vec{z}))(\vec{\sigma}(\vec{y_1}) \otimes \vec{\sigma}(\vec{y_2})) \\
                       &= \vec{\sigma}(\vec{x_1})\vec{\sigma}(\vec{z})\vec{\sigma}(\vec{y_1})\otimes  \vec{\sigma}(\vec{x_2})\vec{\sigma}(\vec{z})\vec{\sigma}(\vec{y_2})
    \end{align*}
    Thus, from Proposition~\ref{prop:involution-facts}~\ref{item:tensor-factor}, the tensor factors separate yielding equations:
    \begin{align*}
        \vec{\sigma}(\vec{x_1})\vec{\sigma}(\vec{z})\vec{\sigma}(\vec{y_1}) = \pm i I_2 \Rightarrow \vec{\sigma}(\vec{x_1})\vec{\sigma}(\vec{z}) = \pm i \vec{\sigma}(\vec{y_1})\\
        \vec{\sigma}(\vec{x_2})\vec{\sigma}(\vec{z})\vec{\sigma}(\vec{y_2}) = \pm i I_2 \Rightarrow \vec{\sigma}(\vec{z})\vec{\sigma}(\vec{x_2}) = \pm i \vec{\sigma}(\vec{y_2})
    \end{align*}
    Combining with similar equations resulting from the last row equation allows us to compute the {(anti)-commutators} $[\vec{\sigma}(\vec{x_i}),\vec{\sigma}(\vec{z})] = 2 i \vec{\sigma}(\vec{y_i})$, $\{\vec{\sigma}(\vec{x_i}),\vec{\sigma}(\vec{z})\} = 0$, $\{\vec{\sigma}(\vec{y_i}),\vec{\sigma}(\vec{z})\}$ for $i \in \{1,2\}$.
    Similar manipulations of the last row and column equations allows us to compute the other (anti)-commutator equations necessary to verify the sets $\{\vec{\sigma}(\vec{x_1}),\vec{\sigma}(\vec{y_1}),\vec{\sigma}(\vec{z})\}$ and $\{\vec{\sigma}(\vec{x_2}),\vec{\sigma}(\vec{y_2}),\vec{\sigma}(\vec{z})\}$ are qubit bases.
\end{proof}
In addition to showing we can express cross product in terms of the magic square, we also need to show that any solution to the magic square can be forced into the form of Lemma~\ref{lem:cross-to-magic}~\ref{item:c2m-magic}.
The following lemma is proved by first by showing that any solution to the magic square is unique up to changing the basis of each qubit, i.e.\ up to unitary maps applied to each qubit.
Therefore, we can diagonalise to the basis of one qubit.
This lemma was inspired by similar uniqueness results in the self-testing literature~\cite{selfTesting,unitaryUnique} though none were precisely suitable for our application.
\begin{lemma}
  \label{lem:magic-unique}
  Let $Q = \Inv(\FC^{4d})$ for $d \geq 1$.
  For every solution $\alpha \in Q^{\mu(\vec{a})}$ to the magic square, there exists a unitary $U\colon \mathbb{C}^2 \rightarrow \mathbb{C}^{2}$ and solution $\alpha_{U} \in Q^{\mu(\vec{a})}$ expressed as:
  \[
    \begin{array}{|c|c|c:}
    \hline
    V(\sigma_x \otimes I_2)V^{-1} \otimes I_d & V(I \otimes \sigma_x)V^{-1} \otimes I_d & V(\sigma_x \otimes \sigma_x)V^{-1}  \otimes I_d\\ \hline
    V(I_2 \otimes \sigma_z )V^{-1} \otimes I_d & V(\sigma_z \otimes I)V^{-1} \otimes I_d & V(\sigma_z \otimes \sigma_z)V^{-1} \otimes I_d \\ \hline
    V(\sigma_x \otimes \sigma_z)V^{-1} \otimes I_d & V(\sigma_z \otimes \sigma_x)V^{-1} \otimes I_d & V(\sigma_y \otimes \sigma_y)V^{-1} \otimes I_d \\ \hline
    \end{array}
  \]
  where $V = U \otimes U$ and $\sigma_x,\sigma_y,\sigma_z$ are the standard Pauli matrices.
\end{lemma}
\begin{proof}
  Given an arbitrary solution $\nu\colon \Var(M) \rightarrow \Inv(\mathbb{C}^{4d})$, we label the resulting operators in the table $A_{ij}$ for the cell in the $i$-th row and $j$-th table.
  Define $X_1 := A_{11}$, $Z_{1} := A_{22}$, $X_2 = A_{12}$, and $Z_{2} := A_{21}$.
  Pairing together $\{X_1,Z_1\}$ and $\{X_2,Z_2\}$, we can make the following observations
  \begin{enumerate}[label=(P\arabic*)]
    \item \label{item:pauli-commute} Every element of $\{X_1,Z_1\}$ commutes with every element of $\{X_2,Z_2\}$.
    \item \label{item:pauli-anti-commute} Within each pair the elements anti-commute, i.e.\ $X_1Z_1 = -Z_1X_1$ and $X_2Z_2 = -Z_2X_2$.
  \end{enumerate}
  Property~\ref{item:pauli-commute} is obvious because each row and column is a pairwise commuting context.
  To prove property~\ref{item:pauli-anti-commute} from the row equations we have that for all $i \in \{1,2,3\}$, $A_{i1}A_{i2}A_{i3} = I$ which implies
  \begin{equation}
    \label{eq:row-equality}
    A_{i3} = A_{i1}A_{i2}
  \end{equation}
  From the last column equation, we have that $A_{13}A_{23}A_{33} = -I$.
  Subsituting the row equation~\eqref{eq:row-equality} into this equation yields:
  \begin{equation}
    \label{eq:last-column-with3}
    (A_{11}A_{12})(A_{21}A_{22})(A_{31}A_{32}) = -I
  \end{equation}
  Subsituting the first two column equations $A_{3j} = A_{1j}A_{2j}$ into~\eqref{eq:last-column-with3} yields:
  \begin{equation}
    \label{eq:last-column-p}
    P = (A_{11}A_{12})(A_{21}A_{22})(A_{11}A_{21})(A_{12}A_{22}) = -I
  \end{equation}
  From this equation, we can conclude that $(A_{11}A_{22})^2 = -I$:
  \begin{align*}
    P &= (A_{11}A_{12})(A_{21}A_{22})(A_{11}A_{21})(A_{12}A_{22}) & \eqref{eq:last-column-p} \\
      &= (A_{11}((A_{12}A_{21})A_{22}))(A_{11}((A_{21}A_{12})A_{22})) \\
      &= (A_{11}A_{22}(A_{12}A_{21}))(A_{11}A_{22}(A_{21}A_{12})) & A_{22} \text{ commutes with both } A_{12},A_{21} \\
      &= (A_{11}A_{22})(A_{12}A_{21})(A_{11}A_{22})(A_{21}A_{12}) \\
      &= (A_{11}A_{22})[(A_{12}A_{21})(A_{11}A_{22})](A_{21}A_{12}) \\
      &= (A_{11}A_{22})(A_{11}A_{22})(A_{12}A_{21})(A_{21}A_{12}) & A_{11}A_{22} \text{ commutes with } A_{12}A_{21} \\
      &= (A_{11}A_{22})^2(A_{12}A_{21})(A_{21}A_{12}) \\
      &= (A_{11}A_{22})^2A_{12}(A_{21}A_{21})A_{12} \\
      &= (A_{11}A_{22})^2 A_{12}(A_{21})^2 A_{12} & A_{21}^2 = I \text{ involution }\\
      &= (A_{11}A_{22})^2 (A_{12})^2 & A_{12}^2 = I \text{ involution }\\
      &= (A_{11}A_{22})^2 \\
    -I &= (A_{11}A_{22})^2 & \eqref{eq:last-column-p}\\
  \end{align*}
  Thus, by Proposition~\ref{prop:involution-facts}\ref{item:unit-anti-commute}, $X_1 = A_{11}$ and $Z_1 = A_{22}$ anti-commute.
  A similar proof demonstrates that $X_2 = A_{12}$ and $Z_2 = A_{21}$ anti-commute establishing property~\ref{item:pauli-anti-commute}.
  From property~\ref{item:pauli-anti-commute}, we can conclude for $j \in \{1,2\}$, the involutions $\{X_j,Z_j,Y_j\}$ where $Y_j := \iu X_j Z_j$ form a qubit basis.
  By property~\ref{item:pauli-commute}, the copies $\langle X_1,Z_1,Y_1 \rangle$ and $\langle X_1,Z_1,Y_1 \rangle$ of the single qubit algebra are represented in different tensor factors over $\mathbb{C}^{4d}$.
  Thus, up to a permutation of tensor factors, we can construct a unitary $V  = V_1 \otimes V_2$ which maps these generators to the standard generators for these two-qubits.
  For the case of $V_1$, we note that the involution $Z_1$ has eigenvectors $\vec{z^+}$ and $\vec{z^-}$ which correspond to the eigenvalues $+1$ and $-1$, respectively.
  The unitary map $V_1\colon \mathbb{C}^2 \rightarrow \mathbb{C}^{2}$ arises from the $2 \times 2$ unitary matrix $M_1$ which maps $\vec{z_+}$ and $\vec{z_-}$ to the standard basis vectors $e_1$ and $e_2$ of $\mathbb{C}^2$.
  We can similarly define the unitary map $V_2$ by considering the eigenvectors of the involution of $Z_2$ to produce matrix $M_2$.
  By construction, for $i \in \{1,2\}$, $V_i$ maps the the eigenvectors of $Z_i$ to $\sigma_z$ via conjugation.
  From this, we obtain that the solution $\alpha$ to $M$ is the table:
  \[
  \begin{array}{|c|c|c|c|}
    \hline
    W(\sigma_x \otimes I_2)W^{-1} \otimes I_d & W(I \otimes \sigma_x)W^{-1} \otimes I_d & W(\sigma_x \otimes \sigma_x)W^{-1}  \otimes I_d\\ \hline
    W(I_2 \otimes \sigma_z )W^{-1} \otimes I_d & W(\sigma_z \otimes I)W^{-1} \otimes I_d & W(\sigma_z \otimes \sigma_z)W^{-1} \otimes I_d \\ \hline
    W(\sigma_x \otimes \sigma_z)W^{-1} \otimes I_d & W(\sigma_z \otimes \sigma_x)W^{-1} \otimes I_d & W(\sigma_y \otimes \sigma_y)W^{-1} \otimes I_d \\ \hline
    \end{array}
  \]
  where $W = V_{1} \otimes V_{2}$.
  Define $U = V_{1}$ and for every $a \in [\vec{a}]$, $\alpha_{U}(a) = I_{2} \otimes UV^{\dagger}_{2}\alpha(a)$ and we get the desired assignment as expressed in the table of the statement.
\end{proof}

Let $t$ be a cross-product term.
For every non-variable subterm $s = r \cross q$, a satisfying assignment of $t$ yields two orthonormal bases.
One basis $B_{s,r}$ containing the pair $\langle s \rangle$,$\langle r \rangle$, and the other containing the pair $\langle s \rangle$,$\langle q \rangle$.
For each such basis $B_{\rho}$ associated to the pair $\rho = (s,r)$, we create a copy $M_{\rho}$ of the magic square:
\begin{equation}
    \label{eq:magic-copy-variable}
    \begin{array}{|c|c|c:}
      \hline
        p_{r,1} & p_{1,r}  & p_{r,r} \\ \hline
        p_{1,s} & p_{s,1} & p_{s,s} \\ \hline
        p_{r,s} & p_{s,r} & p_{r',r'} \\ \hline
    \end{array}
\end{equation}
The indexing of the variables in the copy of $M_{\rho}$ suggest the intended interpetation.
That is, if $t$ is satisfied such that the subterm $s$ resolves to $\mathbb{R}\langle s \rangle = \mathbb{R}\langle r \rangle \cross \mathbb{R}\langle q \rangle$, then e.g.\ $p_{r,s}$ assigned to $\vec{\sigma}(\langle r \rangle) \otimes \vec{\sigma}(\langle s \rangle)$, $p_{s,s}$ is assigned to $\vec{\sigma}(\langle s \rangle) \otimes \vec{\sigma}(\langle s \rangle)$, etc.
Let $\mu_{\rho}(\vec{p_{\rho}})$ be the conjunction of the XOR formulas expressing the equations in $M_{\rho}$ where $[\vec{p_{\rho}}]$ is the set of variables appearing in the table~\eqref{eq:magic-copy-variable}.
Define $\vartheta_{t}(\vec{p})$ to the conjunction of the formulas $\mu_{\rho}(\vec{p_{\rho}})$ over the set of all pairs $\rho = (s,r)$ where $s$ is a non-variable subterm of $t$ and $r$ is an immediate child of $s$.
Since the term satisfiablity problem decides if $\langle t(\mathbb{R}\vec{x_1},\dots,\mathbb{R}\vec{x_1}) \rangle = \mathbb{R}\vec{x_1}$ we identify variables $p_{x_1,1}$, $p_{x_1,x_1}$, $p_{1,x_1}$ with the variables $p_{t,1}$, $p_{t,t}$, $p_{1,t}$, respectively.
\begin{proposition}
\label{prop:xsat-reduction-c4}
Let $g \geq 1$. The following are equivalent:
\begin{enumerate}
  \item $t(x_1,\dots,x_n) \in \XSAT(\RPTwo)$.
  \item $\vartheta_{t}(\vec{p})$ is  in $\strongSAT{\Inv(\mathbb{C}^{4g}})$.
\end{enumerate}
\end{proposition}
\begin{proof}
    Throughout the proof let $Q = \Inv(\mathbb{C}^{4d})$.
    $\Rightarrow$ Suppose $t(x_1,\dots,x_n)$ is satisfied in $\RPTwo$, such that a subterm $s$ of $t$ resolves to $\mathbb{R}\vec{s}$.
    In particular, if $s = r \cross q$, then we know there exist lines such that $\mathbb{R}\vec{s} = \mathbb{R}\vec{r} \cross \mathbb{R}\vec{q}$ in $\RPTwo$.
    We can then express the lines orthogonal to $\vec{r},\vec{s}$ and $\vec{q},\vec{s}$ as $\mathbb{R}\vec{r'} = \mathbb{R}\vec{s} \cross \mathbb{R}\vec{r}$ and $\mathbb{R}\vec{q'} = \mathbb{R}\vec{s} \cross \mathbb{R}\vec{q}$, respectively.
    Hence, by Lemma~\ref{lem:cross-to-magic}, we obtain a satisfying assigment of the equations $M_{\rho}$ for the pair $\rho = (s,r)$ which we render as this table:
    \[
    \begin{array}{|c|c|c:}
      \hline
      \vec{\sigma}(\vec{r}) \otimes I \otimes I_d & I \otimes \vec{\sigma}(\vec{r}) \otimes I_d & \vec{\sigma}(\vec{r}) \otimes \vec{\sigma}(\vec{r}) \otimes I_d\\ \hline
      I \otimes \vec{\sigma}(\vec{s}) \otimes I_d & \vec{\sigma}(\vec{s}) \otimes I \otimes I_d & \vec{\sigma}(\vec{s}) \otimes \vec{\sigma}(\vec{s}) \otimes I_d\\ \hline
    \vec{\sigma}(\vec{r}) \otimes \vec{\sigma}(\vec{s}) \otimes I_d & \vec{\sigma}(\vec{s}) \otimes \vec{\sigma}(\vec{r}) \otimes I_d & \vec{\sigma}(\vec{r'}) \otimes \vec{\sigma}(\vec{r'}) \otimes I_d\\ \hline
    \end{array}
    \]
     and similarly for the pair $\rho' = (s,q)$.
     This yields a meaningful substituion $\alpha_{\rho} \in A^{\mu_{\rho}(\vec{p_{\rho}})}$ such that $\mu_{\rho}^{Q}(\alpha_{\rho}) = 1_{Q}$.
     Observe that if $M_{\rho}$ and $M_{\rho'}$ share a variable $p_z \in [\vec{p_{\rho}}] \cup [\vec{p_{\rho'}}]$, then by construction $\alpha_{\rho}(p_z) = \alpha_{\rho'}(p_z)$.
    Thus, we can obtain a global meaningful subsitution $\alpha \in A^{\vartheta_t(\vec{p})}$ such that $\vartheta_t^{Q}(\alpha) = 1_{Q}$.

    $\Leftarrow$ Conversely, suppose $\vartheta_{t}(\vec{p})$ is strongly satisfied in $Q = \Inv(\mathbb{C}^4)$ via the meaningful substituion $\alpha \in Q^{\vartheta_t(\vec{p})}$.
    In particular, for every subterm $s = r \cross q$ of $t$, we obtain a meaningful substitution $\alpha_{\rho} \in Q^{\mu_{\rho}(\vec{p_{\rho}})}$ for the pair $\rho = (s,r)$ by restriction of $\alpha$.
    Thus, we know the copy $M_\rho$ is satisfied by the involutions in the image of $\alpha_{\rho}$ and by Lemma~\ref{lem:magic-unique}, there is another assignment $\alpha_{\rho,U}$ that is related to the standard assignment up to a unitary $U$.
    Therefore, by Lemma~\ref{lem:cross-to-magic}, we can conclude there exist orthogonal lines $\mathbb{R}\vec{s},\mathbb{R}\vec{r},\mathbb{R}\vec{r'}$ in $\RPTwo$.
    Similarly for the pair $\rho' = (s,q)$, we obtain orthogonal lines $\mathbb{R}\vec{s},\mathbb{R}\vec{q},\mathbb{R}\vec{q'}$.
    Hence, for all subterms $s$ of $t$ we can construct witnesses which satisfy the requisite orthogonality relationships.
    Therefore, $t$ is satisfied in $\RPTwo$.
\end{proof}

\begin{theorem}
\label{thm:fixed-d-complex-sat}
For every $d \geq 4$, $\strongSAT{\Pj(\mathbb{C}^d)}$ is $\ER$-complete.
\end{theorem}
\begin{proof}
  Proposition~\ref{prop:complexdsat-er} proves membership in $\ER$.
  For $\ER$-hardness of the base case $d = 3$, Proposition~\ref{prop:xsat-reduction-c4} yields reductions from the $\ER$-complete problem $\XSAT(\RPTwo)$ to $\strongSAT{\Pj(\FC^4)}$.
  For $\ER$-hardness beyond the base case, we inductively apply Proposition~\ref{prop:dimension-padding} to obtain a reduction from $\strongSAT{\Pj(\mathbb{C}^d)}$ to $\strongSAT{\Pj(\mathbb{C}^{d+1})}$.
\end{proof}
Another consequence of Proposition~\ref{prop:xsat-reduction-c4}, is that we obtain a reduction from $\XSAT(\RPTwo)$ to a restricted version of the satisfiability problem to conjunctions of $\xor$-clauses, i.e.\ XOR-formulas.
The XOR fragment of propositional logic is sometimes called the linear fragment as classically it encodes systems of linear equations over $\FZ_2$.
For a class of pBAs $\Cls$, let $\strongXOR{\Cls}$ denote the decision problem which takes an input an XOR-formula $\varphi(\vec{p})$ and decides if $\varphi(\vec{p})$ is strongly satisfied in some $A \in \Cls$.
Observe that since the formula $\vartheta_t(\vec{p})$ constructed in Proposition~\ref{prop:xsat-reduction-c4} is an XOR-formula, we obtain the following $\ER$-completeness result.

\begin{theorem}
 \label{thm:fixed-d-complex-xorsat}
  For every $g \geq 1$, $\strongXOR{\Pj(\FC^{4g})}$ is $\ER$-complete.
\end{theorem}
Theorem~\ref{thm:fixed-d-complex-xorsat} exhibits an interesting contrast to the classical case.
Classically, $\bothXOR{\Total}$ is in $\mathbf{PTIME}$ and $\SAT{\Total}$ is $\NP$-complete.
By contrast, for the case of dimensions $d$ divisible by $4$, both $\strongXOR{\Pj(\FC^d)}$ and $\strongSAT{\Pj(\FC^d)}$ are both $\ER$-complete.
\begin{remark}
  An inspection of the proofs in~\cite{quantumIsoUndecidable} reveals a dimension preserving reduction from quantum isomorphism of graphs to $\strongXOR{\Pj(\FC^d)}$.
  Thus, a corollary of Theorem~\ref{thm:fixed-d-complex-xorsat} is that for every $d$ divisible by $4$, deciding if two input graphs have a quantum isomorphism witnessed by $d$-dimensional projectors is $\ER$-complete
\end{remark}
\section{Quantum Homomorphisms}
\label{sec:quantum-sat}
We begin this section by introducing a notion of quantum homomorphism between relational structures from~\cite{quantumMonad}.
Quantum homomorphisms correspond to operator solutions of constraint satisfaction problems from~\cite{quantumSchaefer} analogous to the correspondence between ordinary solutions and ordinary homomorphisms between relational structures.
With this notion, we consider the corresponding decision problem $\QHOM(\Hilb)$ and exhibit reductions to/from the satisfiability problem $\SAT{\Pj(\Hilb)}$.
Since reductions preserve the Hilbert space and thus the dimension, we are able to obtain the following corollaries:
\begin{itemize}
  \item $\SAT{\Pj(\FK^{< \omega})}$ is undecidable, where $\Pj(\FK^{< \omega}) = \{\Pj(\FK^{d}) \mid d \in \mathbb{N}\}$ is the class of pBAs arising from finite-dimensional $\FK$-Hilbert spaces.
  \item $\SAT{\Pj(\FK^{\infty})}$ is undecidable,  where $\Pj(\FK^{\infty})$ is the class of pBAs arising from all $\FK$-Hilbert spaces.
  \item For every fixed $d \geq b(\FK)$, $\QHOM(\FK^d)$ is $\ER$-complete where $b(\FR) = 3$ and $b(\FC) = 4$.
\end{itemize}
In the following, we say that a \emph{signature} $\sg$ is a finite set of relationl symbols $R \in \sg$ which each have an associated positive arity $r > 0$.
Given a signature $\sg$, a \emph{$\sg$-structure} $\Ms$ is a set $M$ paired with interpretations $R^{\Ms} \subseteq A^{r}$ for every $r$-ary relational symbol $R \in \sg$.
Given two $\sg$-structures, a \emph{$\sg$-homomorphism} $f\colon \Ms \rightarrow \Ns$ is a set function $M \rightarrow N$ which preserves the interpretations of relations.
The following definition from~\cite{quantumMonad} generalises the notion of $\sg$-homomorphism to the quantum setting.
\begin{definition}
\label{def:ddim-quantum-homomorphism}
Let $\Hilb$ be a real or complex Hilbert space.
A \emph{$\Hilb$-quantum $\sigma$-homomorphism} between two $\sg$-structures $\Ms$ and $\Ns$, denoted $\mathbf{F}\colon \Ms \xrightarrow{\Hilb} \Ns$, is family of projectors $\mathbf{F} = \{F_{m,n}\}_{m \in M,n \in N} \subseteq \Pj(\Hilb)$ satisfying the following conditions:
\begin{enumerate}[label=(QH\arabic*)]
    \item \label{item:normalised} For all $m \in M$, $\sum_{n \in N} F_{m,n} = I_{\Hilb}$.
    \item \label{item:well-defined} If $m,m'$ appear in some relational tuple of $\Ms$ and $n,n' \in \Ns$, then $F_{m,n} \odot_{\Pj(\Hilb)} F_{m',n'}$.
    \item \label{item:non-edge-0} For every $r$-ary symbol $R \in \sg$, $(m_1,\dots,m_r) \in R^{\Ms}$ and $(n_1,\dots,m_r) \not\in R^{\Ns}$, then $F_{m_1,n_1}\dots F_{m_r,n_r} = 0_{\Hilb}$.
\end{enumerate}
\end{definition}
\begin{remark}
A similar notion of quantum homomorphism for undirected simple graphs was orignally introduced in~\cite{quantumHomomorphism}.
Though every quantum homomorphism between graphs in our sense results in a quantum homomorphism in the sense of~\cite{quantumHomomorphism}, the converse does not hold, see e.g.\ \cite{karamlou2025} for details.
\end{remark}
The associated decision problem $\QHOM(\Hilb)$ is the set of pairs $(\Ms,\Ns)$ of finite $\sigma$-structures such that there exists a $\Hilb$-quantum homomorphism $\mathbf{F}\colon \Ms \xrightarrow{\Hilb} \Ns$.
\subsection{Homomorphism to Satisfiability}
In the section, we demonstrate that the standard reduction from $\sigma$-homomorphism to $\SAT{\Total}$ generalises to the quantum setting.
For brevity of presentation, we assume our signature $\sigma$ has a single $r$-ary relation symbol $R$.
For every pair of $\sg$-structures $\Ms,\Ns$, we can rewrite the conditions in Definition~\ref{def:ddim-quantum-homomorphism} as a propositional formula $\phi_{\Ms,\Ns}(\vec{p})$ with variables amongst $[\vec{p}] = \{p_{m,n} \mid m \in M, n \in N\}$.
The formula $\phi_{A,B}(\vec{p})$ is defined as the conjunction $\varphi_{\mathsf{func}}(\vec{p}) \wedge \varphi_{\mathsf{rel}}(\vec{p})$ where
\begin{align*}
  \varphi_{\mathsf{func}}(\vec{p}) &:= \bigwedge_{m \in M} \bigvee_{n \in N} \left( p_{m,n} \wedge \bigwedge_{n \not= n' \in N} \neg p_{m,n'} \right), \\
  \varphi_{\mathsf{rel}}(\vec{p}) &:= \bigwedge_{(m_1,\dots,m_r) \in R^{\Ms}} \bigwedge_{(n_1,\dots,n_r) \not\in R^{\Ns}} \neg (p_{m_1,n_1} \wedge \dots \wedge p_{m_r,n_r}).
\end{align*}
The propositional formula $\phi_{G,H}(\vec{p})$ yields the desired reduction from $\QHOM(\Hilb)$ to $\strongSAT{\Pj(\Hilb)}$ for some Hilbert space $\Hilb$.
\begin{proposition}
\label{prop:qhom-to-sat}
Let $\Ms,\Ns$ be two $\sg$-structures, then $\phi_{\Ms,\Ns}(\vec{p})$ is strongly satisfied in $\Pj(\Hilb)$ iff there exists a $\Hilb$-quantum homomorphism $\mathbf{F}\colon \Ms \xrightarrow{\Hilb} \Ns$.
\end{proposition}
\begin{proof}
  A quantum homomorphism $\mathbf{F} = \{F_{m,n}\}$ corresponds to the assignment $\alpha \in A^{\vec{p}}$ such that $\alpha(p_{m,n}) = F_{m,n}$.
  Unraveling the definitions of $\vee_Q$, $\wedge_Q$, and $\neg_Q$ for $Q = \Pj(\Hilb)$ from Definition~\ref{def:projector-pba}, we note that \ref{item:normalised} is equivalent to $\varphi^{Q}_{\mathsf{func}}(\alpha) = 1_{Q}$, \ref{item:well-defined} corresponds to the fact $\alpha$ is a meaningful substitution, and \ref{item:non-edge-0} is equivalent to $\varphi^{Q}_{\mathsf{rel}}(\alpha) = 1_Q$.
\end{proof}

Let $\QHOM(\FK^{\omega})$ (resp. $\QHOM(\FK^{\infty})$) denote the decision problem consisting of pairs $(\Ms,\Ns)$ for which there exists a $\Hilb$-quantum homomorphism for any finite-dimensional (resp. any) $\FK$-Hilbert space $\Hilb$.
From various results in the literature, we can conclude that $\QHOM(\FK^{< \omega})$ and $\QHOM(\FK^{\infty})$ are undecidable~\cite{quantumSchaefer,quantumIsoUndecidable}, and Proposition~\ref{prop:qhom-to-sat} exhibits a computable reduction from these problems to $\SAT{\Pj(\FK^{< \omega})}$ and $\SAT{\Pj(\FK^{\infty})}$ respectively.
Thus, we obtain the following theorem.
\begin{theorem}
\label{thm:full-dim-sat}
For $\FK \in \{\mathbb{R},\mathbb{C}\}$, $\SAT{\Pj(\FK^{< \omega})}$ and $\SAT{\Pj(\FK^{\infty})}$ are undecidable.
\end{theorem}
In fact, it follows from the proof of Theorem \ref{thm:full-dim-sat}, Proposition~\ref{prop:realdsat-er}, and Proposition~\ref{prop:complexdsat-er} that $\SAT{\Pj(\FK^{< \omega})}$ is complete for the class of recursively enumerable languages.

\newcommand{\X}{\mathcal{X}}
\begin{remark}
  \label{rem:connection-quantum-schaefer}
  A result related to Theorem~\ref{thm:full-dim-sat} was proved in~\cite{quantumSchaefer}.
  Here, they did not use the formulation of partial Boolean algebras, but instead operator assignments for constraint satisfaction problems with Boolean domain $\{-1,1\}$.
  In this formulation, a propositional formula $\phi(p_1,\dots,p_n)$ in CNF form is turned into a constraint satisfaction problem $\X_{\phi} = (X_{\phi},C_{\phi})$ with variables $X = \{x_1,\dots,x_n\}$, and constraints
  \[C_{\phi} = \{(Z_{c},R_{c}) \mid c \text{ clause in } \phi \}, \]
  where each clause $c$ of $\phi$ with variables $p_{i_1},\dots,p_{i_r}$ gives rise to a constraint with scope $Z_{c} = \{x_{i_1},\dots,x_{i_r}\}$ and constraint $R_{c} \in \{-1,+1\}^{r}$ consisting of the set of satisfying tuples of $c$.
  An assignment $f\colon X_{\phi} \rightarrow \{-1,+1\}$ is satisfying of the CSP $\X_{\phi}$ if $(f(x_{i_1}),\dots,f(x_{i_n}) \in R_{c}$ for all clauses $c$.
  This satisfying assignment is in fact a substitution $\phi^{\Two}\colon \Two^{\vec{p}}  \rightarrow \Two$.
  To generalise this to operator assignments, they represent a constraint $(Z,R)$ with scope $Z = \{x_{i_1},\dots,x_{i_r}\}$ as a characteristic multilinear polynomial $P_{R}(x_{i_1},\dots,x_{i_r})$ with the property that $P_{R}(a_{i_1},\dots,a_{i_m}) = -1$ if $(a_{i_1},\dots,a_{i_r}) \in R_{c}$ for $a_{i_z} \in \{-1,1\}$ and $+1$ otherwise.
  An operator assignment for $\Hilb$ is just an assignment $f\colon X_{\phi} \rightarrow \Inv(\Hilb)$ of involutions $\Inv(\Hilb)$ on $\Hilb$ to each element of $X_{\phi}$.
  From such operator assignments, we construct meaningful substitutions that witness the strong satisfiability of $\phi$ in $\Inv(\Hilb)$.
\end{remark}

\subsection{Satisfiability to Homomorphism}
In this section, we show that the standard reduction from $\SAT{\Total}$ to deciding if there exists a $\sg$-homomorphism generalises to the quantum setting.
By Proposition~\ref{prop:cnf-pba}, we can restrict our reduction to the case of propositional formulas in 3CNF form.
Given that the formula $\phi(\vec{p})$ is in CNF form, recall the classical reduction constructs two $\sg$-structures $\Vs_{\phi}$ and $\Ts_{\phi}$ in a signature $\sg_{\phi}$.
The signature $\sg_{\phi}$ has a relational symbol $R_{c}$ for every clause $c$ in $\phi(\vec{p})$ with arity equal to the number of literals, i.e.\ $\leq 3$, in $c$.
The universe of $\Vs_{\phi}$ is the set of variables $[\vec{p}] = \{p_1,\dots,p_n\}$ and for the clause $c = l_{i_1} \vee l_{i_2} \vee l_{i_3}$ where $l_{i_j}$ is a literal in $p_{i_j}$, the interpretation of $R^{\Vs_\phi}_{c}$ is the singleton $\{(p_{i_1},p_{i_2},p_{i_3})\}$.
The universe of $\Ts_{\phi}$ is $\{0,1\}$ and for the clause $c$, $R^{\Ts_\phi}_{c} = \{0,1\}^3 \backslash \{(\partial(l_{i_1}),\partial(l_{i_2}),\partial(l_{i_3}))\}$ where $\delta(l_{i_j}) = 0$ if $l_{i_j}$ is a positive literal and $\delta(l_{i_j}) = 1$ if $l_{i_j}$ is a negative literal.

In order to generalise this construction the quantum case, we simply have to replace $\Ts$ with the structure $\Qd(\Ts)$ where $\Qd$ is a functorial (in fact monadic) construction from~\cite{quantumMonad}.
We recall the $\Qd$ construction here.
For any $\sg$-structure $\Ms$, the universe of $\Qd(\Ms)$ is the set of $M$-labelled PVMs:
\[ \Qd(M) = \{h\colon M \rightarrow \Pj(\FK^d) \mid \sum_{m \in M} h(m) = I_d, \textsf{supp}(h) \text{ is finite }\}  \]
where $\textsf{supp}(h)$ is the set of $m \in M$ such that $h(m) \not= 0_d$.
Elements $h \in \Qd(\Ms)$ can be expressed as formal sums, i.e.\ $\sum_{m \in M} h(m). m$.
For every relational symbol $R \in \sigma$ of arity $r$, $(h_1,\dots,h_r) \in R^{\Qd(\Ms)}$ if
\begin{enumerate}[label=(QR\arabic*)]
    \item \label{item:qr-commute} for all $i,j \in \setr$ and $m,m' \in M$, $h_i(a)\odot_{\Pj(\FK^d)} h_j(m')$.
    \item \label{item:qr-comp} if $(m_1,\dots,m_r) \not\in R^{\Ms}$, then $\prod_{i \in \setr} h_i(m_i) = 0_d$.
\end{enumerate}
The relevant property of $\Qd$ we use is \cite[Proposition 8]{quantumMonad}: there exists an $\FK^d$-quantum $\sg$-homomorphism $\Ms \xrightarrow{\FK^d} \Ns$ if, and only if, there exists an ordinary $\sg$-homomorphism $\Ms \rightarrow \Qd(\Ns)$.

From a 3CNF propositional formula $\phi(\vec{p})$, we construct a signature $\sigma_{\phi}$ and $\sigma_{\phi}$-structures $\Vs_{\phi},\Ts_{\phi}$ satisfying the following proposition.
\begin{proposition}
\label{prop:sat-to-qhom}
Let $\phi(\vec{p})$ be a 3CNF-formula, $\FK \in \{\FR,\FC\}$, $d \geq 1$, and $\Hilb = \FK^d$.
$\Pj(\Hilb)$ satisfies $\phi(\vec{p})$ if, and only if, there is a quantum homomorphism $\mathbf{F}\colon \Vs_{\phi} \xrightarrow{\Hilb} \Ts_{\phi}$.
\end{proposition}
\begin{proof}
  By \cite[Proposition 8]{quantumMonad}, there exists an $\FK^d$-quantum $\sg$-homomorphism $\Ms \xrightarrow{\FK^d} \Ns$ if, and only if, there exists an ordinary $\sg$-homomorphism $\Ms \rightarrow \Qd(\Ns)$.
  Therefore, we proceed by showing that $\Pj(\FK^d)$ satisfies $\phi(\vec{p})$ iff there exists a $\sg$-homomorphism $\sg$-homomorphism $\Ms \rightarrow \Qd(\Ns)$
  $\Rightarrow$ Suppose $Q = \Pj(\FK^d)$ satisfies $\phi(\vec{p})$, then there exists a meaningful substitution $\alpha \in A^{\phi(\vec{p})}$ such that $\phi^{Q}(\alpha) = I_d$.
From $\alpha$, we can define a homomorphism $f_{\alpha}\colon \Vs_\phi \rightarrow \Qd(\Ts_\phi)$ where $f_{\alpha}(p_i)$ is the PVM expressed by the formal sum $\alpha(p_i) . 1 + (I_d - \alpha(p_i)) . 0$.
We must show that $f_v(p_i)$ is homomorphism.
Suppose $(p_{i_1},\dots,p_{i_r}) \in R^{\Vs_\phi}_{c}$ and we aim to show that $(f_{v}(p_{i_1}),\dots,f_{v}(p_{i_r})) \in R^{\Qd(\Ts_\phi)}_{c}$.
Since $(p_{i_1},\dots,p_{i_r}) \in R^{\Vs_\phi}_{c}$, by construction, $c = l_{i_1} \vee \dots \vee l_{i_r}$ is a clause of $\phi(p_1,\dots,p_n)$ where $l_{i_z}$ is literal in variable $p_{i_z}$.
By $Q^{\phi(\vec{p})}$ being a meaningful domain and $Q^{\phi(\vec{p})} \subseteq Q^{c(\vec{p})}$, $\alpha(p_{i_z})$ commutes with $\alpha(p_{i_w})$ for all $w,z \in \setr$, thus $(f_{v}(p_{i_1}),\dots,f_{v}(p_{i_r}))$ satisfies condition~\ref{item:qr-commute}.
To show $(f_{v}(p_{i_1}),\dots,f_{v}(p_{i_r}))$ satisfies condition~\ref{item:qr-comp}, we note that only tuple that is not in $R^{\Ts_{\phi}}_{c}$ is $ (\partial(l_{i_1}),\dots,\partial(l_{i_r}))$ where $\partial(l_{i_z}) = 0$ if $l_{i_z}$ is a positive literal and $\partial(l_{i_z}) = 1$ if $l_{i_z}$ is a negative literal.
Thus, $f_{v}(p_{i_z})(\partial(l_{i_z}))$ is $I - \alpha(p_{i_z})$ if $l_{i_z}$ is positive and $f_{v}(p_{i_z})(\partial(l_{i_z}))$ is $\alpha(p_{i_z})$ if $l_{i_z}$ is negative.
Since $c = l_{i_1} \vee \dots \vee l_{i_r}$ is satisfied, $m(\neg c) = m( \neg l_{i_1}  \wedge \dots \wedge \neg l_{i_r}) = 0_d$, so it must be the case that $\prod_{z \in \setr} f_{v}(p_{i_z})(\partial(l_{i_z})) = 0_{d}$.

$\Leftarrow$ Given a $\sigma_{\phi}$-homomorphism $f\colon V_{\phi} \rightarrow \Qd(\Ts_{\phi})$, we can define a meaningful substitution $\alpha \in Q^{\phi(\vec{p})}$ where $\alpha(p_i) = f(p_i)(1)$.
Condition~\ref{item:qr-commute}, and $f$ being a $\sigma$-morphism ensures that $v \in Q^{c(\vec{p})}$ for all clauses $c \in \phi$.
Condition~\ref{item:qr-comp}, applied to $R^{\Ts_{\phi}}_{c}$, ensures that $(\neg c)^{Q}(\alpha) = 0_d$ for every clause $c$ of $\phi$ and thus $\phi^{Q}(\alpha) = I_d$.
\end{proof}

\begin{theorem}
  Let $b(\FR) = 3$ and $b(\FC) = 4$.
  For every $d \geq b(\FK)$, $\QHOM(\FK^d)$ is $\ER$-complete.
\end{theorem}
\begin{proof}
  Theorem~\ref{thm:fixed-d-real-sat} and Theorem~\ref{thm:fixed-d-complex-sat} demonstrate that $\strongSAT{\Pj(\FK^d)}$ is an $\ER$-complete problem.
  Proposition~\ref{prop:qhom-to-sat} is a polynomial-time reduction from $\QHOM(\FK^d)$ to $\strongSAT{\Pj(\FK^d)}$.
  Proposition~\ref{prop:sat-to-qhom} is a polynomial-time reduction from $\strongCNF{\Pj(\FK^d)}$ and $\strongSAT{\Pj(\FK^d)}$ (by Proposition~\ref{prop:cnf-pba}) to $\QHOM(\FK^d)$.
\end{proof}

\section{Conclusion}
The Cook-Levin theorem is a classical result of computational complexity that tells us that the problem of deciding whether a propositional formula is satisfiable is $\NP$-complete.  Or, dually, the class of propositional tautologies is $\cNP$-complete.   What the Kochen-Specker theorem tells us is that if we interpret the propositions not as classical truth values but as measurement outcomes in a quantum system, then not every classical tautology is always true.  Equivalently, there are propositional formulas that are satisfiable in such systems that are not classically satisfiable.  What is then the complexity of the class of satisfiable formulas?  This is the question that we set out to address.

When we interpret satisfiability to mean satisfiable in some non-trivial partial Boolean algebra, the problem is again $\NP$-complete.  The hardness is a direct consequence of the $\NP$-hardness of classical satisfiability but the upper bound is non-trivial to establish as the natural witness to satisfiability can be of doubly exponential size.  Our main contribution here, in Theorem~\ref{thm:allsat}, is to construct a suitable polynomial-size witness that can be efficiently verified.

If we restrict ourselves to the partial Boolean algebras that motivated the question, namely those of projectors on a Hilbert space, the picture is more complicated.  We show that for any finite dimension $d \geq 4$, the problem of deciding satisfiability in such a projector space of dimension $d$ is complete for the existential theory of the reals.  However, the problem of determining whether a formula is satisfiable in the algebra of projectors of some finite-dimensional Hilbert space is undecidable.
The $\ER$-completeness result situates this problem in an interesting place in the complexity landscape and relates it to a number of natural problems in computational geometry.  This opens up avenues of further research relating quantum complexity classes to such problems.

\printbibliography

\end{document}